\documentclass[submission]{dmtcs} 
\usepackage[latin1]{inputenc}
\usepackage{amsmath,amsfonts,amssymb,bbm,url,color,tikz,mathrsfs,verbatim,rotating,subfigure,array}
\usetikzlibrary{arrows,shapes,snakes,automata,backgrounds,petri}

\newcommand{\acco}[1]{\left\{#1\right\}}

\newcommand{\pa}[1]{\left(#1\right)}

\newcommand{\incl}{\subseteq}

\newcommand{\joliN}{\mathcal{N}}

\newcommand{\f}{\mathfrak{f}}
\newcommand{\g}{\mathfrak{g}}
\newcommand{\s}{\mathfrak{s}}

\newcommand{\N}{\mathbb{N}}
\newcommand{\R}{\mathbb{R}}

\newcommand{\Z}{\mathbb{Z}}

\newcommand{\id}{\operatorname{id}}

\newcommand\sac{\sqsubseteq}
\newcommand{\simu}{\preccurlyeq}
\newcommand\grp[2]{{#1}^{<#2>}}
\newcommand\bloc[1]{b_{#1}}
\newcommand\debloc[1]{b^{-1}_{#1}}

\newcommand{\Ac}{F} 
\newcommand{\Ad}{G}
\newcommand{\globA}{F} 
\newcommand{\globAd}{G} 
\newcommand\ZZ{\Z}
\newcommand{\Alphh}{\mathcal{Q}}

\newcommand{\autobis}[1]{$\operatorname{Spot}[#1]$}
\newcommand{\joliM}{\mathscr{M}}
\newcommand{\ord}{\operatorname{ord}}

\newcommand{\Id}{\mathbbm{1}}

\newcommand\mir{M}
\newtheorem{defi}{Definition}
\newtheorem{thm}{Theorem}
\newtheorem{prop}{Proposition}
\newtheorem{ppty}{Property}
\newtheorem{lemma}{Lemma}
\newtheorem{cor}{Corollary}

\author{Vincent Nesme\addressmark{1}
  \and Guillaume Theyssier\addressmark{2}\thanks{Research partially supported by project ANR EMC NT09 555297 (French national research agency)}}

\address{\addressmark{1}Freie Universit\"at Berlin\\
  \addressmark{2}LAMA (CNRS, Universit\'e de Savoie),\newline Campus Scientifique 73376 Le Bourget-du-Lac Cedex, France}

\title[Selfsimilarity, Simulation and Spacetime Symmetries]{Selfsimilarity, Simulation and Spacetime Symmetries}

\revised{\today}


\begin{document}

\maketitle

 \begin{abstract} 
   We study intrinsic simulations between cellular automata and introduce a new necessary condition for a CA to simulate another one. Although expressed for general CA, this condition is targeted towards surjective CA and especially linear ones. Following the approach introduced by the first author in an earlier paper, we develop proof techniques to tell whether some linear CA can simulate another linear CA. Besides rigorous proofs, the necessary condition for the simulation to occur can be heuristically checked via simple observations of typical space-time diagrams generated from finite configurations. As an illustration, we give an example of linear reversible CA which cannot simulate the identity and which is 'time-asymmetric', i.e. which can neither simulate its own inverse, nor the mirror of its own inverse.
  \end{abstract}

\section{Introduction and definitions}

Cellular automata (CA) are well-known for the variety of behaviors they can exhibit. A lot of classification schemes where proposed in the literature, trying to make this variety of behaviors more intelligible \cite{Wolfram:1984:CTCA,Gilman:1987:CLA,Kurka97}. Such classifications usually consist in a (finite) list of distinctive properties giving raise to a partition of the class of all CA. Another approach consists in defining a simulation relation between CA, and studying the ordered structure induced by the simulation. We follow this latter approach, and more precisely the simulation relation $\simu_i$ defined in \cite{bulking1,bulking2} giving rise to the notion of intrinsic universality \cite{surveyOllinger}. The intuition behind this simulation relation is simple: a CA is simulated by another if some rescaling of the first is a sub-automaton of a rescaling of the second.

More formally, we restrict ourselves to dimension 1 and the definition
is as follows.  A CA $\Ac$ is a \emph{sub-automaton} of a CA $\Ad$,
denoted ${\Ac\sac\Ad}$, if there is an injective map $\varphi$ from
$A$ to $B$ (state sets of $\Ac$ and $\Ad$ respectively) such that
${\overline{\varphi}\circ\globA=\globAd\circ \overline{\varphi}}$,
where ${\overline{\varphi}:A^\ZZ\rightarrow B^\ZZ}$ denotes the
uniform extension of $\varphi$ to configurations.  We sometimes write
${\Ac\sac_{\varphi}\Ad}$ to make $\varphi$ explicit. This definition
is standard but yields a very limited notion of simulation: a given
CA can only admit a finite set of (non-isomorphic) CA as sub-automata.
Therefore, following works of J.~Mazoyer, I.~Rapaport and
N.~Ollinger \cite{rap,ollingerphd,bulking1,bulking2}, we will add rescaling
operations to the notion of simulation. The ingredients of rescaling
operations are simple: packing cells into blocks, iterating the rule
and composing with a translation (formally, we use shift CA
$\sigma_z$, $z\in\ZZ$, whose global rule is given by $\sigma_z (c)_x =
c_{x-z}$ for all $x\in\ZZ$).  Given any state set $\Alphh$ and any $m\geq
1$, we define the bijective packing map ${\bloc{m}: \Alphh^\ZZ\rightarrow
  \bigl(\Alphh^m\bigr)^\ZZ}$ by:
\[\forall z\in\ZZ : \bigl(\bloc{m}(c)\bigr)(z) = \bigl(c(mz),\ldots,c(mz+m-1)\bigr)\]
for all ${c\in \Alphh^\ZZ}$. The rescaling $\grp{\Ac}{m,t,z}$ of $\Ac$ by
parameters $m$ (packing), ${t\geq 1}$ (iterating) and ${z\in\ZZ}$
(shifting) is the CA of state set $\Alphh^m$ and global rule:
\[\bloc{m} \circ \sigma_z\circ \globA^t \circ \debloc{m}.\]
With these definitions, we say that $\Ac$ simulates $\Ad$, denoted
${\Ad\simu\Ac}$, if there are rescaling parameters $m_1$, $m_2$,
$t_1$, $t_2$, $z_1$ and $z_2$ such that
${\grp{\Ad}{m_1,t_1,z_1}\sac\grp{\Ac}{m_2,t_2,z_2}}$.  

Determining whether some given CA simulates another given CA is hard (undecidable in general \cite[section 4.3]{bulking2}). For instance, looking at typical space time diagrams of two CA gives no clue on whether one simulates another, because the simulation can occur on a set of configurations of measure 0. Despite the general undecidability of the simulation relation, one can still hope to better understand its restriction to some specific classes of CA. For instance, the simulation relation is fully understood on products of shifts \cite[theorem 3.4]{bulking2} thanks to a 'characteristic sequence' which is essentially the sequence of ratio of translation vectors. Hence, if ${F=\sigma_0\times\sigma_1\times\sigma_3}$, one can prove that $F$ cannot simulate $F^{-1}=\sigma_0\times\sigma_{-1}\times\sigma_{-3}$ because they do not have the same characteristic sequence. 

In this paper, we introduce a general necessary condition for a simulation between two CA to be possible. It focuses on surjective CA, but we will essentially use it on linear reversible CA. This condition is expressed as a characteristic set $\chi$ of points of the real half-plane which is decreasing w.r.t. $\simu$ (Theorem~\ref{thm:simu} below): \[F\simu G\Rightarrow \chi(G)\subseteq\chi(F).\] A striking property of $\chi$ is that it can be somewhat visualized on typical space time diagrams of linear CA. Moreover, the set $\chi$ is closely related to so-called 'Green functions' of linear CA for which systematic analysis techniques have been developed in \cite{fractal}. Hence, formal proofs of impossibility of simulation between two linear CA can be derived from heuristic observations of space-time diagrams in a quasi-automatic way. 

The set of reversible CA is somewhat structured with respect to $\simu$ since it possesses a maximal element (\textit{i.e.,} a reversible universal CA \cite[theorem 4.5]{bulking2}) and verifies the following \cite[theorem 4.4]{bulking2}:
\[F\simu G \Rightarrow F^{-1}\simu G^{-1}\]
Therefore, a reversible CA is either $\simu$-equivalent to its inverse, or $\simu$-incomparable to it. The most complex reversible CA, reversible universal CA, are all $\simu$-equivalent to their own inverse. Coming back to the example $F$ above (product of shifts), we have that $F$ and $F^{-1}$ are $\simu$-incomparable.  Following~\cite{block}, let us associate to every reversible CA $F$ its dual $\tilde F = \mir\circ F^{-1}\circ \mir$, where $\mir$ is the mirror transformation on configurations (${\mir(c)_z=c_{-z}}$).  Any product of shifts is self-dual, and generally speaking it seems to be hard to come up with CA that do not simulate their dual, while non-time-symmetric CA in the sense of \cite{gamo} come in profusion.

An interesting question in this context is how different a reversible CA can be from its dual.  As an illustration of the necessary condition for simulation between CA that is given by Theorem~\ref{thm:simu}, we study in section~\ref{sec:linear} some reversible linear CA.  The first one simulates its inverse, its mirror, its dual, but not the identity;  the second one simulates neither the identity nor its inverse or its mirror image or its dual.

\section{Simulation and geometry}
\label{sec:condition}

The basic ingredient in this section is the collection of functions telling how a change of value of the center cell in the initial configuration will affect some other cell's value at some step in the future. Such functions are often studied for linear cellular automata (see section~\ref{sec:linear}) and are sometimes called 'Green functions' in this context \cite{moore1998}.

Let $\Ac$ be any CA and fix some $x\in\Z$ and some $y\in\N$. For any configuration $c\in \Alphh^\Z$ and any $q\in \Alphh$, we denote by $\phi_c(q)$ the following configuration:
\[\phi_c(q)_z =
\begin{cases}
  q &\text{ if }z=0,\\
  c(z) &\text{ else.}
\end{cases}\]
We then denote by $\Ac_{x,c}^y: \Alphh\rightarrow \Alphh$ the map ${q\mapsto \left(\Ac^y(\phi_c(q))\right)_x}$.

For instance, if $F$ is simply the identity, ${\globA}^{y}_{x,c}$ is the identity when $x=0$, otherwise it is the constant function $q\mapsto c(x)$.  For a less trivial example, consider the cas where $F=\bigoplus$ is the sum with neighborhood $\acco{0,-1}$ over $\Alphh=\Z/2\Z$, i.e. $\bigoplus(c)_x=c(x)+c(x-1)$.  Starting from a single nonzero cell, iterations of this automaton generate Pascal's triangle modulo $2$.  For $x\in\N$, let $x=\sum_{n\in\N} b_x(n)2^{n}$, with $b_x(n)\in\acco{0,1}$, be its binary representation, and $B_x=\acco{n\in\N | b_x(n)=1}$.  Then 

$$\bigoplus\nolimits^y_{x,c}(q)=\left\{\begin{array}{ll}\bigoplus\nolimits^y(c)_x+q & \text{if $x\geq 0$ and $B_x\incl B_y$}\\ \bigoplus\nolimits^y(c)_x &\text{else}\end{array}\right..$$

We are interested in positions in space-time where the influence of the center cell is concentrated, whatever the initial configuration (see figure~\ref{fig:spot}).

\begin{defi}
$\Ac$ has the property \autobis{x,y,l,r} for $x\in\Z$ and $y,l,r\in\N$ if

\begin{itemize}
\item ${\globA}^{y}_{x,c}$ is a bijection for all configurations $c$; and
\item ${\globA}^{y}_{z,c}$ is a constant function for all $c$ and all ${z\in [x-l;x+r]\setminus \acco{x}}$.
\end{itemize}
\end{defi}

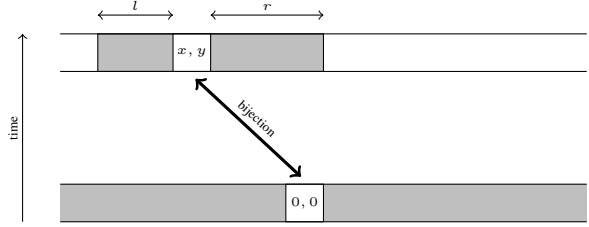
\begin{figure}[htbp]
  \tiny
  \centering
  \begin{tikzpicture}
    \fill[fill=gray!50!white] (-3,0)--(4,0)--(4,.5)--(-3,.5)--cycle;
    \draw[fill=white] (0,0)--(.5,0)--(.5,.5)--(0,.5)--cycle;
    \draw (.25,.25) node (origin) {$0,0$};
    \draw (-3,0)--(4,0);
    \draw (-3,.5)--(4,.5);
    \draw[fill=gray!50!white] (-2.5,2)--+(3,0)--+(3,.5)--+(0,.5)--cycle;
    \draw[fill=white] (-1.5,2)--+(.5,0)--+(.5,.5)--+(0,.5)--cycle;
    \draw[<->] (-2.5,2.75)--+(1,0) node[midway,above] {$l$};
    \draw[<->] (-1,2.75)--+(1.5,0) node[midway,above] {$r$};
    \draw (-1.25,2.25) node (dest) {$x,y$};
    \draw (-3,2)--(4,2);
    \draw (-3,2.5)--(4,2.5);
    \draw[very thick,<->] (.2,.6)--(-1.2,1.9) node [midway,sloped,above] {bijection};
    \draw[->] (-3.5,0)--+(0,2.5) node[midway,sloped,above] {time};
  \end{tikzpicture}
  \caption{Property \autobis{x,y,l,r}. Gray zones correspond to cells whose state does not change (either fixed in the initial configuration or kept constant in the $y$-th iteration of the CA).}
  \label{fig:spot}
\end{figure}

$\bigoplus$ thus fulfills \autobis{x,y,l,r} if and only if, for any $z\in [x-l;x+r]$, $B_z\incl B_y$ is equivalent to $x=z$ (considering that ``$B_z\incl B_y$'' is a false statement when $B_z$ is undefined).

\begin{lemma}\label{lemma:motsfinis}
If $F$ has the property \autobis{x,y,l,r}, then $F^y(\Alphh^\Z)$ contains all the words of size $\max(l,r)+1$.
\end{lemma}
\begin{proof}	
Let us suppose, without loss of generality, that $l$ is no larger than $r$. 
Let $\bar q=(q_0,\ldots,q_r)\in \Alphh^{r+1}$.  We are going to construct $c\in \Alphh^\Z$ such that $F^y(c)_{0,\ldots,r}=\bar q$.  Start with an arbitrary $c\in \Alphh^\Z$.  We can first modify $c_{r-x}$ in such way that $F^y(c)_r=q_r$; then we can change $c_{r-x-1}$, on which $F^y(c)_r$ does not depend, so that $F^y(c)_{r-1}=q_{r-1}$; and so on, until we choose $c_{-x}$, on which $F(c)_{1,\ldots,r}$ does not depend, so that $F(c)_{0}=q_0$.  
\end{proof}

For instance, $\bigoplus$ fulfills \autobis{0,1,+\infty,0}, which implies that it must be surjective.

\begin{lemma}\label{lemma:xx'}
Let $\joliN^{y}$ be the neighborhood of $\Ac^{y}$.  If $\Ac$ fulfills \autobis{x,y,l,r} and \autobis{x',y',l',r'} with $[x'-l';x'+r']+\joliN^{y'}\incl [-l;r]$, then \autobis{x+x',y+y',l',r'} also holds.  
\end{lemma}

\begin{proof}
By definition of the neighborhood, $[x'-l';x'+r']+\joliN^{y'}\incl [-l;r]$ implies that ${\globA}^{y'}(c)_{[x'-l';x'+r']}$ is a function of $c_{[-l;r]}$.  Applying that to $c=\sigma_{-x}\circ {\globA}^{y}(d)$, we get that the restriction of ${\globA}^{y+y'}(d)$ to ${[x+x'-l';x+x'+r']}$ is a function of ${\globA}^{y}(d)_{[x-l;x+r]}$.  Thus we get:
\begin{itemize}
\item ${\globA}^{y+y'}(d)_{[x+x'-l';x+x'+r']\setminus\acco{x+x'}}$ does not depend on $d_0$;
\item ${\globA}^{y+y'}(d)_{x+x'}$ depends only on ${\globA}^{y}(d)_{[x-l;x+r]}$, which in turn, according to \autobis{x,y,l,r}, depends injectively on $d_0$.
\end{itemize}
\end{proof}

The central idea of the paper is to study the set of parameters $(x,y,l,r)$ for which the property \autobis{x,y,l,r} holds, and use that set to obtain necessary conditions for simulations between cellular automata. However, we won't use the set of parameters directly because the simulation relation is invariant by space-time rescalings and this set is not. Instead we will look at 'scale-free' structures inside this set of parameters. More precisely, given some integer $p$, we look for infinite geometric progressions of order $p$ in the set of parameters. Hence we obtain a kind of fingerprint for each CA which is well-behaved with respect to space-time transformations involved in the simulation relation (Theorem~\ref{thm:simu} below). Moreover, as shown by examples developed latter in this paper, this fingerprint is closely related to the self-similar structure observed in typical space-time of some linear CA.
Technically, this is how the definition goes.


\begin{defi}
For a CA $\Ac$ and an integer $p\geq 2$, we denote ${X_p(\Ac)}$ the set of points ${(x,y)\in\R\times [0;+\infty)}$ such that for some $k\in\N$, for every large enough $n\in\N$, $\Ac$ fulfills \autobis{x p^n,y p^n, p^{n-k}, p^{n-k}}.
\end{defi}

$X_2(\bigoplus)$ is for instance the set of points $(x,y)\in\R\times [0;+\infty)$ that can be written $x=\frac{a}{2^n}$ and $y=\frac{b}{2^n}$ with $a,b\in\N$ and $B_a\incl B_b$: its restriction to $\R\times [0;1]$ is the dyadic part of a (shifted) Sierpi\'nski triangle.

It can be noted that $X_p$ is necessarily of measure $0$, and is self-similar, since by Lemma~\ref{lemma:xx'} every point of $X_p$ is the tip of a small copy of $X_p$ within itself.  One can also notice that if $F$ is not surjective, then $X_p(F)$ is reduced to the singleton $\acco{(0,0)}$.  Indeed, if $X_p(F)$ is not reduced to a singleton, then according to Lemma~\ref{lemma:motsfinis}, the image of $F$ contains every finite word, which implies, by compactness, that $F$ is surjective.



We now detail, in a series of properties, how $X_p$ is modified under the action of the transformations involved in the simulation of a CA by another.  First, the shift.  Let $\s_z$ be the transformation of the plane defined by $\s_z(x,y)=\pa{x+zy,y}$.  The following property is obvious, by definition of $X_p$.

\begin{ppty}\label{ppte:shift}
$X_p(\sigma_z\circ\Ac)=\s_z\pa{X_p(\Ac)}$.
\end{ppty}

Let us now consider iteration and grouping.  Let $\g_t$ be the transformation of the plane defined by $\g_t(x,y)=\pa{x,\frac{y}{t}}$: notice that $\g_p\pa{X_p(\Ac)}=X_p(\Ac)$.  Let $\f_m$ be the transformation of the plane defined by $f_m(x,y)=\pa{\frac{x}{m},y}$.



\begin{ppty}\label{ppte:denseiter}
$X_p(\Ac^t)$ is a dense subset of $\g_t\pa{X_p(\Ac)}$. 
\end{ppty}
\begin{proof}
The inclusion is immediate from the definition.  What might be slightly less immediate is why these sets are not obviously equal.  Given the definition, $X_p(\Ac)$ must be included in $\R_p^2$, where $\R_p$ is the set of real numbers having finite $p$-adic expansion.  Actually, we do have $\g_t \pa{X_p(\Ac)}\cap \R^2_p= X_p(\Ac^t)$, so the equality without the intersection is certainly true if $t\in\R_p$, not quite so in general. Let us now prove the density.

Let $(x,y)\in X_p(\Ac)$.  We want to find a sequence $(x_n,y_n)$ of points of $X_p(\Ac)$ converging to $(x,y)$ such that for all $n$, $y_n\in t\R_p$.  For a finite sequence of integers $0=i_{n,0}<i_{n,1}<\ldots<i_{n,l}$ ($l$ is a constant independent of $n$ to be fixed later), we define $\eta_n=\sum\limits_{j=0}^{l} p^{-i_{n,j}}$ and $(x_n,y_n)=\eta_n (x,y)$.  We have three requirements:

\begin{itemize}
\item $(x_n,y_n)$ must converge to $(x,y)$: it is sufficient to have $\lim\limits_{n\to+\infty}i_{n,1}=+\infty$
\item $(x_n,y_n)$ must be an element of $X_p(\Ac)$.   This is guaranteed as long as $i_{n,j+1}-i_{n,j}$ is always large enough.  More precisely, by definition of $X_p$ there exists some $k$ such that for every large enough integer $n$, $\Ac$ fulfills \autobis{xp^n,yp^n,p^{n-k},p^{n-k}}.  Therefore, if $i_{n,l}-i_{n,l-1}$ is large enough (depending on $k$ and the neighborhood of $F$), we get from Lemma~\ref{lemma:xx'} that $(1+p^{-i_{n,l}+i_{n,l-1}})(x,y)$ is in $X_p(\Ac)$.  By recursion on $l$, we get ultimately $(x_n,y_n)\in X_p(\Ac)$.
\item $y_n$ must be in $t\R_p$, which means the integer $p^{i_{n,l}}\sum\limits_{j=0}^{l} p^{-i_{n,j}}$ must be a multiple of $t$.
\end{itemize}

So, it all boils down to finding increasing integer sequences $0=i_0<i_1<\cdots<i_{l}$ where $i_{j+1}-i_j$ is arbitrary large, and such that $t$ divides $p^{i_{l}}\sum\limits_{j=0}^{l} p^{-i_{j}}$.  That is clearly possible: the sequence of powers of $p$ is ultimately periodic modulo $t$, so if we choose the $i_j$-s spaced by multiples of this period and $l=t$, we can easily meet the conditions.
\end{proof}

\begin{ppty}\label{ppte:densegp}
  $X_p(b_m\circ \Ac\circ b_{m}^{-1})$ is a dense subset of $\f_m\pa{X_p(\Ac)}$.
\end{ppty}
\begin{proof}
  Let $\Ad=b_m\circ \Ac\circ b_{m}^{-1}$. Let $x\in\Z$ and $y,l,r\in\N$ with $l\geq 1$ and $r\geq 1$.
  First, it follows from definitions that, for any configuration $c$ of $\Ac$, $\Ad_{x,b_m(c)}^y$ is constant if and only if, for all ${z\in\{mx-m+1,\cdots,mx+m-1\}}$, $\Ac_{z,c}^y$ is constant. Moreover $\Ad_{x,b_m(c)}^y$ bijective implies $\Ac_{mx,c}^y$ bijective. This shows that if $\Ad$ has property \autobis{x,y,l,r} then $\Ac$ has property \autobis{mx,y,ml,mr}.

  Now suppose that $\Ac$ has property \autobis{mx,y,ml,mr} and fix some configuration $c$ of $F$. Then it is straightforward to check that $\Ad_{x,b_m(c)}^y$ is bijective (because it sends each component of $\Alphh^m$ on itself) and $\Ad_{z,b_m(c)}^y$  is constant for any $z\in[x-l;x+r]\setminus\acco{x}$.
  
We have shown that $\Ac$ has property \autobis{mx,y,ml,mr} if and only if $\Ad$ has property \autobis{x,y,l,r}. Thus we have
\[(x,y)\in X_p(\Ad)\iff (mx,y)\in X_p(\Ac).\]
This implies ${X_p(\Ad)\subseteq\f_m\pa{X_p(\Ac)}}$. To prove the density, it is sufficient to prove that ${X_p(\Ac)\cap m\R_p}$ is dense in $X_p(\Ac)$ which can be done using the same argument as in the proof of property~\ref{ppte:denseiter}.
\end{proof}

It only remains to consider the case of the sub-automaton. 

\begin{ppty}\label{ppte:sub}
If $\Ad\sac\Ac$ then $X_p(\Ac)\incl X_p(\Ad)$.
\end{ppty}
\begin{proof}
  It is straightforward to check that if $\Ac$ has property
  \autobis{x,y,l,r} then so does $\Ad$.  The property follows.
\end{proof}

Properties \ref{ppte:shift}, \ref{ppte:denseiter}, \ref{ppte:densegp} and \ref{ppte:sub} prove the following theorem. 

\begin{thm}
\label{thm:simu}
If $\Ac$ simulates $\Ad$, then there exist rational numbers $\beta$ and $\alpha,\gamma>0$ such that for every integer $p\geq 2$, $\pi_{\alpha,\beta,\gamma}(\overline{X_p(\Ac)})\incl\overline{X_p(\Ad)}$, where $\pi_{\alpha,\beta,\gamma}(x,y)=(\alpha x+\beta y, \gamma y)$.
\end{thm}

The determination of $X_p$ is not easy in general, but the following basic facts can be established straightforwardly from the definitions:
\begin{itemize}
\item if $F$ is a shift, then $\overline{X_p(F)}$ is a line passing through the origin;
\item if $F$ is nilpotent (\textit{i.e.} $\exists t$ s.t. $F^t$ is a constant function), then $X_p(F)=\{(0,0)\}$;
\item ${X_p(F\times G) = X_p(F)\cap X_p(G)}$.
\end{itemize}

Theorem~\ref{thm:simu} above shows that $X_p(\Ac)$ represent obstructions for $\Ac$ to simulate other CA: the bigger $X_p(\Ac)$ is, the smaller the family of CA $\Ac$ can simulate. Using the basic facts above, we can give some concrete formulations of this intuition.

\begin{cor}Let $p\geq 2$ be an integer and $\Ac$ a CA. Then we have:
  \begin{itemize}
  \item If $\Ac$ simulates the identity, then $\overline{X_p(\Ac)}$ must be included in a line passing through the origin;
  \item If $\Ac$ is intrinsically universal, then $X_p(\Ac)=\{(0,0)\}$;
  \item If $\Ac$ is reversible universal (i.e. it can simulate any reversible CA), then $X_p(\Ac)=\{(0,0)\}$;
  \end{itemize}
\end{cor}
\begin{proof}
  All items use Theorem~\ref{thm:simu}.
  Item 1 and 2 are direct consequences of the computation of $X_p$ for the identity and nilpotent CA (an intrinsically universal CA must simulate any nilpotent CA). Item 3 uses the fact that a reversible universal CA must simulate $\sigma\times\sigma^{-1}$ whose $X_p$ is a singleton.
\end{proof}


The purpose of the next section is to focus on a class of CA that generally have more interesting $X_p$: linear cellular automata.

\section{Linear Cellular Automata}
\label{sec:linear}

\begin{figure}[htbp]
\includegraphics[width=\textwidth]{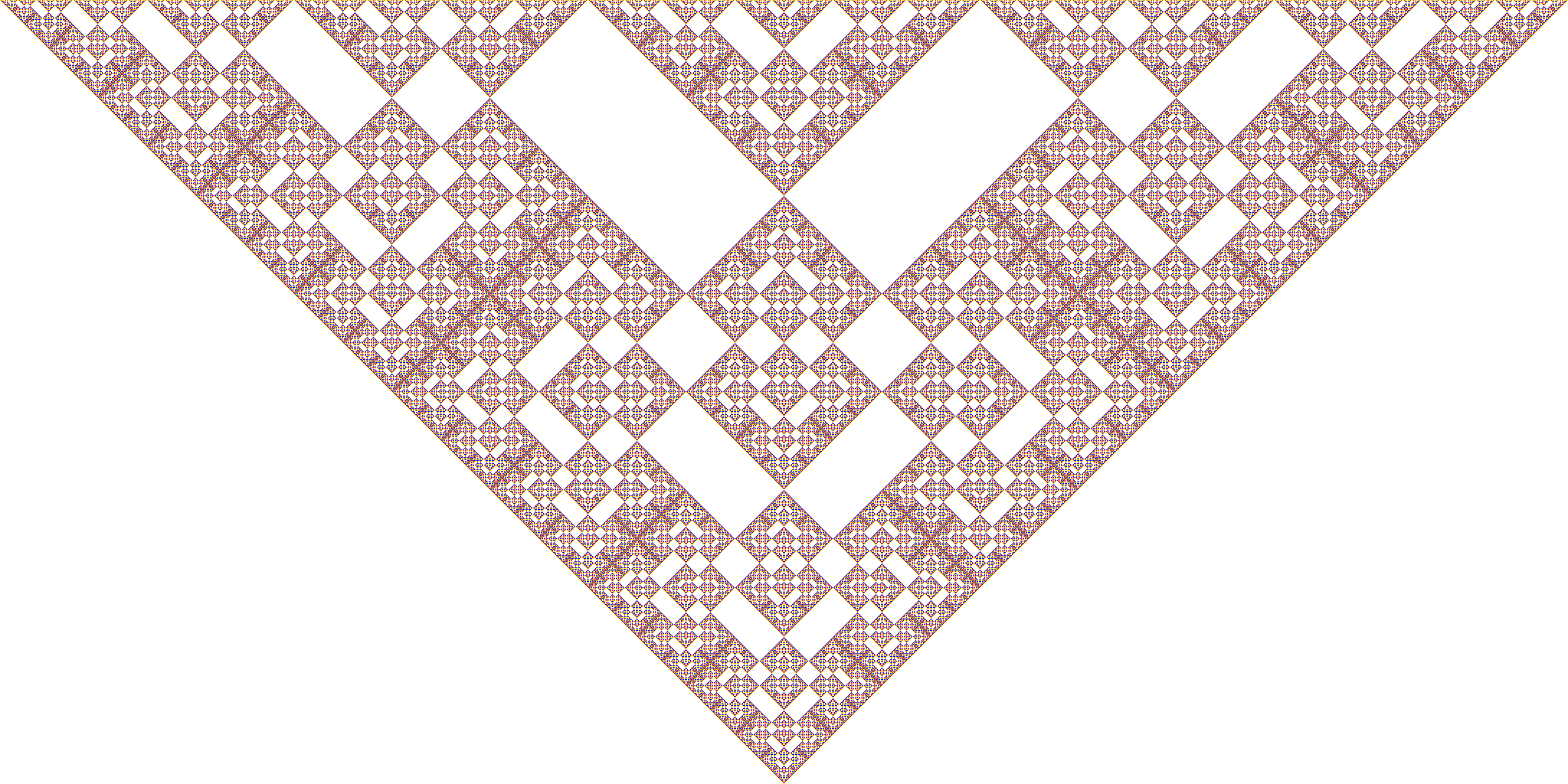}
\caption{Spacetime diagram of $\Theta$ up to a large power of $2$. Also $\overline{X_2(\Theta)}$. Time goes from bottom to top\label{fig:theta}}
\end{figure}

More often than not, one can get a good idea about what $\overline{X_p}$ looks like just by examining the spacetime diagram.  We think in particular of linear CA in the sense of \cite{fractal}.  In this case,  $\Alphh=R^d$, where $R$ is a finite abelian ring, and $d$ some positive integer.  The algebra of CA that are homomorphisms of $\pa{R^d}^{\Z}$ is then isomorphic to $\joliM_d(R)[u,u^{-1}]$: read section~1 of \cite{fractal} for details.  

If $\Ac$ is such a linear CA and if $0$ denotes the neutral element of $R^d$, the sets $X_p$ can be derived from the functions $\Ac_{x,\overline{0}}^y$ where $\overline{0}$ denotes the uniform configuration everywhere equal to $0$. Indeed, for any configuration $c$, we have:
\[\Ac_{x,c}^y\text{ bijective (resp. constant)}\iff \Ac_{x,\overline{0}}^y\text{ bijective (resp. constant)}\]

In the sequel we denote $\Ac_{x,\overline{0}}^y$ by $\Ac_x^y$.  The remainder of this section focuses on reversible cellular automata.

\subsection{$\Theta$: a reversible CA which cannot simulate the identity}

Let us look at a more interesting example. The alphabet is now $\pa{\Z_2}^2$, and the transition is given by 
\[\Theta=\pa{\begin{array}{cc}0  & 1 \\ 1 & u^{-1}+1+u\end{array}}.\]  Since it already serves as a red thread through \cite{fractal}, we will pass very quickly on it.  Let us notice here that, since its determinant is $1$, it is reversible, and that its inverse is $\Theta^{-1}=\pa{\begin{array}{cc}u^{-1}+1+u & 1 \\ 1 & 0 \end{array}}$.  Obviously, $\Theta$ simulates its own inverse: in fact $\Theta^{-1}\sac_{\varphi}\Theta$ with $\varphi=\pa{\begin{array}{cc}0&1\\1&0\end{array}}$.

Figure~\ref{fig:theta} represents the spacetime diagram of $\Theta$ up to a large power of $2$, for an initial configuration consisting of one single nonzero cell.  $\Theta$ is ``well-behaved'' in the sense that these spacetime diagrams, for increasingly large powers of $2$, converge to $\overline{X_2(\Theta)}$.  It thus gives in a sense a purely visual proof of the fact that $\Theta$ does not simulate the identity.  Of course, this requires actually some background knowledge, in order for the proof to be correct.  One must know that $\Theta$ is a linear CA, and that $\overline{X_p}$ actually corresponds to its limit spacetime diagram, or at least is not limited to one line.  While $X_p$ is not defined in \cite{fractal}, the information given there on the way to describe the limit spacetime diagram by means of a substitution system justifies this assertion.  The crucial point is that any block that is not empty contains a reduced copy of the whole pattern, which means that in the neighborhood of any non-white point in the limit spacetime diagram, there is a copy of the whole thing, whose tip is then a point in $X_2(\Theta)$; therefore $X_2(\Theta)$ is dense in this pattern. 
And so, adding that $\Theta$ is its own mirror image, we get:

\begin{prop}
$\Theta$ simulates its mirror, its inverse and its dual, but cannot simulate the identity.
\end{prop}

\subsection{$\Gamma$: a life in pictures}

Let us now provide an example of a CA that is both space- and time-asymmetric, in the sense that it cannot simulate any of the CA derived from it by inverting space and/or time.  This will be 
\[\Gamma=\pa{\begin{array}{ccc} 0&0&1\\ 0&1&u\\ 1&u&0 \\ \end{array}}\in \joliM_3(\Z_2)[u,u^{-1}].\]  Its inverse is given by $\Gamma^{-1}=\pa{\begin{array}{ccc} u^2&u&1\\ u&1&0\\ 1&0&0 \\ \end{array}}$.  We are going to give only the proof that it does not simulate its inverse: the proof of the two other results would add only length to this article, and can surely be left as an exercise to the reader.

\begin{figure}[htbp]
\begin{center}
    \subfigure[$X_2(\Gamma)$]{\includegraphics[width=0.3\textwidth]{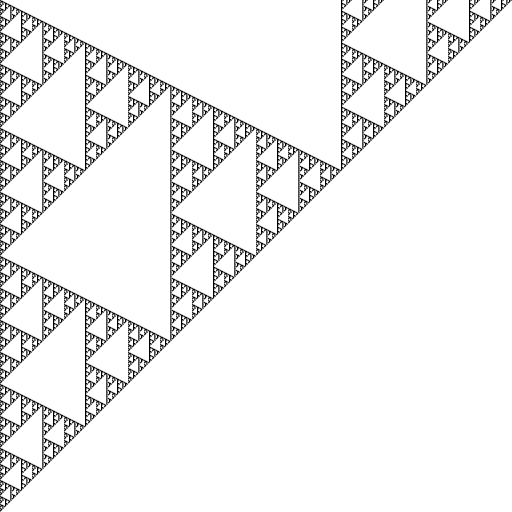}\label{fig:gamma1}}
   \hfil
    \subfigure[$X_2(\Gamma^{-1})$]{\includegraphics[width=0.6\textwidth]{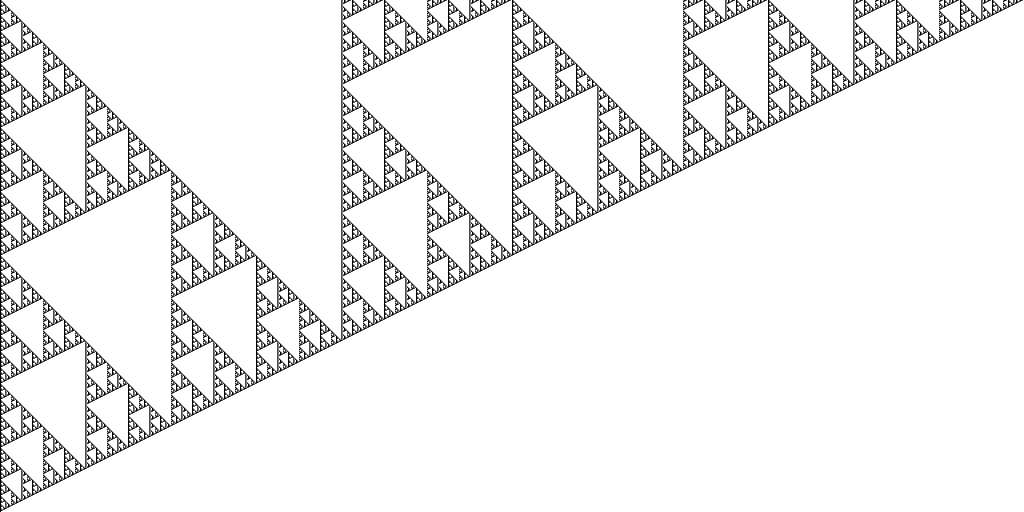}\label{fig:gamma2}}
    \caption{$X_2$ with the second coordinate restricted to $[0,1]$ (time goes from bottom to top)\label{fig:gamma}}
\end{center}
\end{figure}

Let us imagine for one blissful moment that we know $\overline{X_2}$ to be accurately represented by Figures \ref{fig:gamma1} and \ref{fig:gamma2} (actually these figures are mirror images of spacetime diagrams up to a large power of $2$).  How do we conclude then?  

Supposing that $\Gamma$ simulates $\Gamma^{-1}$, we know from theorem~\ref{thm:simu} that for some $\alpha,\beta,\gamma$, $\pi_{\alpha,\beta,\gamma}(\overline{X_2(\Gamma)})$ should be included in $\overline{X_2(\Gamma^{-1})}$.  Since there are only two lines passing through the origin in $\overline{X_2(\Gamma^{-1})}$, $\pi_{\alpha,\beta,\gamma}$ must send respectively $\R (0,1)$ and $\R (-1,1)$ on $\R (0,1)$ and $\R (-2,1)$, which implies $\beta=0$.  Now if we consider the lines joining these two axes, they have slope $\frac{1}{2}$ for $\Gamma$, $1$ in $\Gamma^{-1}$, which means $\alpha=2\gamma$.  So, if $\Gamma$ simulates its inverse, $\overline{X_2(\Gamma)}$ should be, modulo a change of scale, included into $\overline{X_2(\Gamma^{-1})}$, which is clearly not the case.

To make this proof rigorous, we need a tool to prove properties of
$X_2$ for $\Gamma$ and $\Gamma^{-1}$. We are going to follow section~3 of \cite{fractal}, which gives a procedure to derive, from the transition matrix of the CA, a substitution system generating the Green functions (see Proposition~4 of \cite{fractal}). More precisely, we will associate to each CA $F$ a $2\times 2$ substitution system, that is a finite set $E$ and a function ${e:\Z\times\N\rightarrow E}$ such that:
\begin{itemize}
\item $F_x^y$ is a function of $e(x,y)$;
\item for $i,j\in\{0,1\}$, $e(2x+i,2y+j)$ is a function of $e(x,y)$ and $i$ and $j$.
\end{itemize}

The next two subsections give the substitution systems for $\Gamma$ and $\Gamma^{-1}$, and subsection~\ref{sub:final} uses them to formally prove negative result concerning simulation.

\subsubsection{A substitution system for $\Gamma$}
\label{sec:gamma}

The minimal polynomial of $\Gamma$ is $X^3+X^2+(1+u^2)X+1$, so we have the following recurrence relation. 

\begin{equation}
\forall x\in\Z\forall n,y\in\N\; \qquad y<3 \cdot 2^{n}\Longrightarrow
\label{eq:gamma}
\Gamma^{3\cdot 2^n+y}_x=\Gamma^{2^{n+1}+y}_x+\Gamma^{2^n+y}_x+\Gamma^{2^n+y}_{x-2^{n+1}}+\Gamma^y_x
\end{equation}

Now we define $\alpha_{j}(x,y)$ in the following way: these are the coefficients in $\Z/2\Z$ such that for every function $(x,y)\mapsto\Xi_x^y$ fulfilling equation (\ref{eq:gamma}) in lieu of $\Gamma$,

\begin{equation}
\Xi^y_x = \sum\limits_{j=0}^{2}\sum\limits_{i\in\Z} \alpha_{j}(x-i,y)\Xi_{i}^{j}.
\end{equation}

For every $x\in\Z$, $y\in \N$ and $s,t\in\acco{0,1}$, we have 

\begin{equation}
\Xi_{2x+s}^{2y+t}=\sum\limits_{i\in\Z} \alpha_{0}(x-i,y)\Xi_{2i+s}^t + \alpha_{1}(x-i,y)\Xi_{2i+s}^{2+t}+\alpha_{2}(x-i,y)\Xi_{2i+s}^{4+t}.
\end{equation}

In the case $s=t=0$, we have the following derivation:

\begin{equation}
\begin{array}{rcl}
\Xi_{2x}^{2y} &= &\sum\limits_i \alpha_{0}(x-i,y)\Xi_{2i}^0 + \alpha_{1}(x-i,y)\Xi_{2i}^{2}+\alpha_{2}(x-i,y) \Xi_{2i}^4 \\
& = &\sum\limits_i \alpha_{0}(x-i,y)\Xi_{2i}^0 + \alpha_{1}(x-i,y)\Xi_{2i}^{2}+\alpha_{2}(x-i,y) \pa{\Xi_{2i-2}^2 + \Xi_{2i-2}^1 + \Xi_{2i}^0} \\
&=& \sum\limits_i \pa{\alpha_{0}(x-i,y)+\alpha_2 (x-i,y)}\Xi_{2i}^0 +\alpha_{2}(x-1-i,y)\Xi_{2i}^{1}\\&&\hfill + \pa{\alpha_{1}(x-i,y)+\alpha_2 (x-1-i,y)}\Xi_{2i}^{2} \\
\end{array}
\end{equation}

which is to be compared with the definition of $\alpha_{j}$:

\begin{equation}
\Xi_{2x}^{2y}= \sum\limits_i\sum_{j=0}^2 \alpha_{j}(2x-i,2y)\Xi_{i}^j.
\end{equation}

The comparison shows that $\alpha_{j}(2x,2y)$  is a function of $\alpha_{j}(x-i,y)$ for some values of $i$.  $\Gamma$ is peculiar in that $\alpha_{j}(2x+1,2y)=0$, which simplifies our work.  The same operation now has to be performed for $\alpha_{j}(2x,2y+1)$.

\begin{equation}
\begin{array}{rcl}
\Xi_{2x}^{2y+1} &= &\sum\limits_i \alpha_{0}(x-i,y)\Xi_{2i}^1 + \alpha_{1}\Xi_{2i}^3 +\alpha_{2}(x-i,y)\Xi_{2i}^5\\
 &= &\sum\limits_i \alpha_{0}(x-i,y)\Xi_{2i}^1 + \alpha_{1}(x-i,y)\pa{\Xi_{2i}^2 + \Xi_{2i}^1 + \Xi_{2i-2}^1 + \Xi_{2i}^0} \\&&\hfill +\alpha_{2}(x-i,y) \pa{\Xi_{2i}^1 + \Xi_{2i-2}^1 + \Xi_{2i-4}^1 + \Xi_{2i-2}^0}\\
&=&  \sum\limits_i \pa{ \alpha_{1}(x-i,y)+\alpha_{2}(x-1-i,y)}\Xi_{2i}^0 \\
& &  + (\alpha_{0}(x-i,y)+\alpha_{1}(x-1-i,y)+\alpha_{1}(x-i,y)+\alpha_{2}(x-2-i,y)\\&&\hfill +\alpha_{2}(x-1-i,y)+\alpha_{2}(x-i,y))\Xi_{2i}^1 + \alpha_{1}(x-i,y)\Xi_{2i}^2 \\
\end{array}
\end{equation}

Using the representation $\begin{array}{|c|}\hline\alpha_{2}\\ \alpha_{1}\\ \alpha_{0}\\ \hline \end{array}$, we get the following substitution.

\begin{displaymath}
\begin{array}{c}
\begin{array}{|c|}
  \hline
  \alpha_{\cdot}(x,y)\\
  \hline
\end{array}\\
\downarrow\\
\begin{array}{|c|c|}
  \hline
  \alpha_{\cdot}(2x,2y+1) & \alpha_{\cdot}(2x+1,2y+1)\\
  \hline
  \alpha_{\cdot}(2x,2y) & \alpha_{\cdot}(2x+1,2y)\\
  \hline
\end{array}\\
\rotatebox[origin=c]{90}{=}\\
\scalebox{0.85}
{$\begin{array}{|c|c|}
  \hline
  \begin{array}{c} \alpha_{1}(x,y) \\ \alpha_{0}\pa{x,y}+\alpha_{1}\pa{x-1,y}+\alpha_{1}\pa{x,y}+\alpha_{2}\pa{x-2,y}+ \alpha_{2}\pa{x-1,y}+\alpha_{2}\pa{x,y}\\ \alpha_{1}(x,y)+  \alpha_{2}(x-1,y)\end{array} & \begin{array}{c} 0\\ 0 \\ 0 \\ \end{array}\\
  \hline
  \begin{array}{c}\alpha_{1}(x,y)+ \alpha_2(x-1,y) \\ \alpha_2(x-1,y)\\ \alpha_{0}\pa{x,y}+\alpha_{2}\pa{x,y}\\ \end{array} &  \begin{array}{c}0\\ 0\\0 \\ \end{array}\\
  \hline
\end{array}$}
\end{array}
\end{displaymath}

This needs some grouping; for instance, in the present situation, the substitution scheme uses $\alpha_1(x-1,y)$, which is not an information contained in the initial cell.  For instance, if we want to determine $\alpha_0(2x,2y+t)$ for $t\in\acco{0,1}$, we need to know $\alpha_0(x,y)$, $\alpha_1(x,y)$, $\alpha_2(x,y)$ and $\alpha_2(x-1,y)$.  The smallest grouping that will allow us to carry all that information is 

$$\begin{array}{|cccc|}
  \hline
  \alpha_{2}(x-3,y)&\alpha_{2}(x-2,y)&\alpha_{2}(x-1,y)&\alpha_{2}(x,y) \\
&\alpha_{1}(x-2,y)&\alpha_{1}(x-1,y)&\alpha_{1}(x,y) \\
& & \alpha_{0}(x-1,y)& \alpha_{0}(x,y) \\
  \hline
\end{array}.$$  

This gives us an alphabet of size $2^9=512$, and the substitution scheme is

\begin{center}
\scalebox{1}
{$\begin{array}{c}
 \begin{array}{|c|}
    \hline
    \begin{array}{cccc}
      a & b & c & d\\
      & e & f & g \\
      && h & i \\
    \end{array}\\
    \hline
  \end{array}
\\
  \downarrow
\\
  \begin{array}{|cccc|cccc|}
    \hline
    
      0 & f & 0 & g & f & 0 & g & 0  \\
      & a+b+c+e+f+h &  0  & b+c+d+f+g+i && 0  & b+c+d+f+g+i & 0  \\
      & &  0 & c+g&& &  c+g & 0  \\ 

    \hline

      0 & b+f & 0 & c+g & b+f & 0 & c+g & 0 \\
      & b & 0 & c & &  0 & c &0  \\
      && 0 & d+i &  & & d+i &0  \\
    \hline
  \end{array}
\end{array}
$}
\end{center}

The initial state for this substitution system is $\begin{array}{c|c|c|c|c}\hline \cdots&0&D&0&\cdots\\ \hline\end{array}$, and to a cell $ \begin{array}{|c|}
    \hline
    \begin{array}{cccc}
      a & b & c & d\\
      & e & f & g \\
      && h & i \\
    \end{array}\\
    \hline
  \end{array}$ in position $(x,y)$ corresponds the Green function $\Gamma_x^y=\pa{\begin{array}{ccc}d+i & c &g \\c &b+d+i+g & c+f\\g &c+f &b+d+i \\ \end{array}}$.

For a letter $x$ in $\acco{a,b,\ldots,i}$, let $\begin{array}{|c|}\hline X\\ \hline \end{array}$ denote the cell where $x$ has the value 1 whereas all other letters are set to $0$.  We can notice that $A$, $E$ and $H$ are completely equivalent: they all substitute to $\begin{array}{|c|c|}\hline E & 0 \\ \hline 0 & 0 \\ \hline \end{array}$, and project onto $0$ in the computation of $\Gamma_x^y$.  We can therefore simplify this system a bit by putting $A=E=H=0$:

\newcommand\metsmoicaici[3]{\draw (#1,#2) node[draw,minimum size=.4cm] () {#3};}

\begin{figure}[htbp]
\begin{center}
  \fontsize{4}{4}\selectfont
  \begin{tikzpicture}[xscale=.4,yscale=.4]
    \metsmoicaici{31}{31}{F}\metsmoicaici{30}{31}{BG}\metsmoicaici{29}{31}{C}\metsmoicaici{27}{31}{F}\metsmoicaici{26}{31}{DG}\metsmoicaici{23}{31}{F}\metsmoicaici{22}{31}{BG}\metsmoicaici{21}{31}{C}\metsmoicaici{2}{31}{BD}\metsmoicaici{1}{31}{CF}\metsmoicaici{0}{31}{G}
\metsmoicaici{30}{30}{B}\metsmoicaici{29}{30}{C}\metsmoicaici{28}{30}{BD}\metsmoicaici{27}{30}{CF}\metsmoicaici{26}{30}{DG}\metsmoicaici{22}{30}{B}\metsmoicaici{21}{30}{C}\metsmoicaici{4}{30}{BD}\metsmoicaici{3}{30}{CF}\metsmoicaici{2}{30}{BDG}\metsmoicaici{1}{30}{C}\metsmoicaici{0}{30}{D}
\metsmoicaici{29}{29}{F}\metsmoicaici{28}{29}{G}\metsmoicaici{27}{29}{F}\metsmoicaici{26}{29}{DG}\metsmoicaici{22}{29}{BD}\metsmoicaici{21}{29}{C}\metsmoicaici{6}{29}{BD}\metsmoicaici{5}{29}{CF}\metsmoicaici{4}{29}{G}\metsmoicaici{3}{29}{F}\metsmoicaici{2}{29}{DG}\metsmoicaici{1}{29}{F}\metsmoicaici{0}{29}{G}
\metsmoicaici{28}{28}{B}\metsmoicaici{27}{28}{CF}\metsmoicaici{26}{28}{DG}\metsmoicaici{24}{28}{BD}\metsmoicaici{23}{28}{CF}\metsmoicaici{22}{28}{BDG}\metsmoicaici{21}{28}{C}\metsmoicaici{8}{28}{BD}\metsmoicaici{7}{28}{CF}\metsmoicaici{6}{28}{BDG}\metsmoicaici{5}{28}{C}\metsmoicaici{4}{28}{B}\metsmoicaici{3}{28}{CF}\metsmoicaici{2}{28}{DG}\metsmoicaici{0}{28}{D}
\metsmoicaici{27}{27}{F}\metsmoicaici{26}{27}{BG}\metsmoicaici{25}{27}{CF}\metsmoicaici{24}{27}{G}\metsmoicaici{23}{27}{F}\metsmoicaici{22}{27}{BG}\metsmoicaici{21}{27}{C}\metsmoicaici{10}{27}{BD}\metsmoicaici{9}{27}{CF}\metsmoicaici{8}{27}{G}\metsmoicaici{7}{27}{F}\metsmoicaici{6}{27}{BG}\metsmoicaici{5}{27}{C}\metsmoicaici{3}{27}{F}\metsmoicaici{2}{27}{BG}\metsmoicaici{1}{27}{CF}\metsmoicaici{0}{27}{G}
\metsmoicaici{26}{26}{B}\metsmoicaici{25}{26}{C}\metsmoicaici{24}{26}{D}\metsmoicaici{22}{26}{B}\metsmoicaici{21}{26}{C}\metsmoicaici{12}{26}{BD}\metsmoicaici{11}{26}{CF}\metsmoicaici{10}{26}{BDG}\metsmoicaici{9}{26}{C}\metsmoicaici{8}{26}{D}\metsmoicaici{6}{26}{B}\metsmoicaici{5}{26}{C}\metsmoicaici{2}{26}{B}\metsmoicaici{1}{26}{C}\metsmoicaici{0}{26}{D}
\metsmoicaici{25}{25}{F}\metsmoicaici{24}{25}{G}\metsmoicaici{22}{25}{BD}\metsmoicaici{21}{25}{C}\metsmoicaici{14}{25}{BD}\metsmoicaici{13}{25}{CF}\metsmoicaici{12}{25}{G}\metsmoicaici{11}{25}{F}\metsmoicaici{10}{25}{DG}\metsmoicaici{9}{25}{F}\metsmoicaici{8}{25}{G}\metsmoicaici{6}{25}{BD}\metsmoicaici{5}{25}{C}\metsmoicaici{1}{25}{F}\metsmoicaici{0}{25}{G}
\metsmoicaici{24}{24}{B}\metsmoicaici{23}{24}{CF}\metsmoicaici{22}{24}{BDG}\metsmoicaici{21}{24}{C}\metsmoicaici{16}{24}{BD}\metsmoicaici{15}{24}{CF}\metsmoicaici{14}{24}{BDG}\metsmoicaici{13}{24}{C}\metsmoicaici{12}{24}{B}\metsmoicaici{11}{24}{CF}\metsmoicaici{10}{24}{DG}\metsmoicaici{8}{24}{B}\metsmoicaici{7}{24}{CF}\metsmoicaici{6}{24}{BDG}\metsmoicaici{5}{24}{C}\metsmoicaici{0}{24}{D}
\metsmoicaici{23}{23}{F}\metsmoicaici{22}{23}{BG}\metsmoicaici{21}{23}{C}\metsmoicaici{18}{23}{BD}\metsmoicaici{17}{23}{CF}\metsmoicaici{16}{23}{G}\metsmoicaici{15}{23}{F}\metsmoicaici{14}{23}{BG}\metsmoicaici{13}{23}{C}\metsmoicaici{11}{23}{F}\metsmoicaici{10}{23}{DG}\metsmoicaici{7}{23}{F}\metsmoicaici{6}{23}{BG}\metsmoicaici{5}{23}{C}\metsmoicaici{2}{23}{BD}\metsmoicaici{1}{23}{CF}\metsmoicaici{0}{23}{G}
\metsmoicaici{22}{22}{B}\metsmoicaici{21}{22}{C}\metsmoicaici{20}{22}{BD}\metsmoicaici{19}{22}{CF}\metsmoicaici{18}{22}{BDG}\metsmoicaici{17}{22}{C}\metsmoicaici{16}{22}{D}\metsmoicaici{14}{22}{B}\metsmoicaici{13}{22}{C}\metsmoicaici{12}{22}{BD}\metsmoicaici{11}{22}{CF}\metsmoicaici{10}{22}{DG}\metsmoicaici{6}{22}{B}\metsmoicaici{5}{22}{C}\metsmoicaici{4}{22}{BD}\metsmoicaici{3}{22}{CF}\metsmoicaici{2}{22}{BDG}\metsmoicaici{1}{22}{C}\metsmoicaici{0}{22}{D}
\metsmoicaici{21}{21}{F}\metsmoicaici{20}{21}{G}\metsmoicaici{19}{21}{F}\metsmoicaici{18}{21}{DG}\metsmoicaici{17}{21}{F}\metsmoicaici{16}{21}{G}\metsmoicaici{13}{21}{F}\metsmoicaici{12}{21}{G}\metsmoicaici{11}{21}{F}\metsmoicaici{10}{21}{DG}\metsmoicaici{5}{21}{F}\metsmoicaici{4}{21}{G}\metsmoicaici{3}{21}{F}\metsmoicaici{2}{21}{DG}\metsmoicaici{1}{21}{F}\metsmoicaici{0}{21}{G}
\metsmoicaici{20}{20}{B}\metsmoicaici{19}{20}{CF}\metsmoicaici{18}{20}{DG}\metsmoicaici{16}{20}{D}\metsmoicaici{12}{20}{B}\metsmoicaici{11}{20}{CF}\metsmoicaici{10}{20}{DG}\metsmoicaici{4}{20}{B}\metsmoicaici{3}{20}{CF}\metsmoicaici{2}{20}{DG}\metsmoicaici{0}{20}{D}
\metsmoicaici{19}{19}{F}\metsmoicaici{18}{19}{BG}\metsmoicaici{17}{19}{CF}\metsmoicaici{16}{19}{G}\metsmoicaici{11}{19}{F}\metsmoicaici{10}{19}{DG}\metsmoicaici{3}{19}{F}\metsmoicaici{2}{19}{BG}\metsmoicaici{1}{19}{CF}\metsmoicaici{0}{19}{G}
\metsmoicaici{18}{18}{B}\metsmoicaici{17}{18}{C}\metsmoicaici{16}{18}{D}\metsmoicaici{12}{18}{BD}\metsmoicaici{11}{18}{CF}\metsmoicaici{10}{18}{DG}\metsmoicaici{2}{18}{B}\metsmoicaici{1}{18}{C}\metsmoicaici{0}{18}{D}
\metsmoicaici{17}{17}{F}\metsmoicaici{16}{17}{G}\metsmoicaici{14}{17}{BD}\metsmoicaici{13}{17}{CF}\metsmoicaici{12}{17}{G}\metsmoicaici{11}{17}{F}\metsmoicaici{10}{17}{DG}\metsmoicaici{1}{17}{F}\metsmoicaici{0}{17}{G}
\metsmoicaici{16}{16}{B}\metsmoicaici{15}{16}{CF}\metsmoicaici{14}{16}{BDG}\metsmoicaici{13}{16}{C}\metsmoicaici{12}{16}{B}\metsmoicaici{11}{16}{CF}\metsmoicaici{10}{16}{DG}\metsmoicaici{0}{16}{D}
\metsmoicaici{15}{15}{F}\metsmoicaici{14}{15}{BG}\metsmoicaici{13}{15}{C}\metsmoicaici{11}{15}{F}\metsmoicaici{10}{15}{DG}\metsmoicaici{2}{15}{BD}\metsmoicaici{1}{15}{CF}\metsmoicaici{0}{15}{G}
\metsmoicaici{14}{14}{B}\metsmoicaici{13}{14}{C}\metsmoicaici{12}{14}{BD}\metsmoicaici{11}{14}{CF}\metsmoicaici{10}{14}{DG}\metsmoicaici{4}{14}{BD}\metsmoicaici{3}{14}{CF}\metsmoicaici{2}{14}{BDG}\metsmoicaici{1}{14}{C}\metsmoicaici{0}{14}{D}
\metsmoicaici{13}{13}{F}\metsmoicaici{12}{13}{G}\metsmoicaici{11}{13}{F}\metsmoicaici{10}{13}{DG}\metsmoicaici{6}{13}{BD}\metsmoicaici{5}{13}{CF}\metsmoicaici{4}{13}{G}\metsmoicaici{3}{13}{F}\metsmoicaici{2}{13}{DG}\metsmoicaici{1}{13}{F}\metsmoicaici{0}{13}{G}
\metsmoicaici{12}{12}{B}\metsmoicaici{11}{12}{CF}\metsmoicaici{10}{12}{DG}\metsmoicaici{8}{12}{BD}\metsmoicaici{7}{12}{CF}\metsmoicaici{6}{12}{BDG}\metsmoicaici{5}{12}{C}\metsmoicaici{4}{12}{B}\metsmoicaici{3}{12}{CF}\metsmoicaici{2}{12}{DG}\metsmoicaici{0}{12}{D}
\metsmoicaici{11}{11}{F}\metsmoicaici{10}{11}{BG}\metsmoicaici{9}{11}{CF}\metsmoicaici{8}{11}{G}\metsmoicaici{7}{11}{F}\metsmoicaici{6}{11}{BG}\metsmoicaici{5}{11}{C}\metsmoicaici{3}{11}{F}\metsmoicaici{2}{11}{BG}\metsmoicaici{1}{11}{CF}\metsmoicaici{0}{11}{G}
\metsmoicaici{10}{10}{B}\metsmoicaici{9}{10}{C}\metsmoicaici{8}{10}{D}\metsmoicaici{6}{10}{B}\metsmoicaici{5}{10}{C}\metsmoicaici{2}{10}{B}\metsmoicaici{1}{10}{C}\metsmoicaici{0}{10}{D}
\metsmoicaici{9}{9}{F}\metsmoicaici{8}{9}{G}\metsmoicaici{6}{9}{BD}\metsmoicaici{5}{9}{C}\metsmoicaici{1}{9}{F}\metsmoicaici{0}{9}{G}
\metsmoicaici{8}{8}{B}\metsmoicaici{7}{8}{CF}\metsmoicaici{6}{8}{BDG}\metsmoicaici{5}{8}{C}\metsmoicaici{0}{8}{D}
\metsmoicaici{7}{7}{F}\metsmoicaici{6}{7}{BG}\metsmoicaici{5}{7}{C}\metsmoicaici{2}{7}{BD}\metsmoicaici{1}{7}{CF}\metsmoicaici{0}{7}{G}
\metsmoicaici{6}{6}{B}\metsmoicaici{5}{6}{C}\metsmoicaici{4}{6}{BD}\metsmoicaici{3}{6}{CF}\metsmoicaici{2}{6}{BDG}\metsmoicaici{1}{6}{C}\metsmoicaici{0}{6}{D}
\metsmoicaici{5}{5}{F}\metsmoicaici{4}{5}{G}\metsmoicaici{3}{5}{F}\metsmoicaici{2}{5}{DG}\metsmoicaici{1}{5}{F}\metsmoicaici{0}{5}{G}
\metsmoicaici{4}{4}{B}\metsmoicaici{3}{4}{CF}\metsmoicaici{2}{4}{DG}\metsmoicaici{0}{4}{D}
\metsmoicaici{3}{3}{F}\metsmoicaici{2}{3}{BG}\metsmoicaici{1}{3}{CF}\metsmoicaici{0}{3}{G}
\metsmoicaici{2}{2}{B}\metsmoicaici{1}{2}{C}\metsmoicaici{0}{2}{D}
\metsmoicaici{1}{1}{F}\metsmoicaici{0}{1}{G}
\metsmoicaici{0}{0}{D}

  \end{tikzpicture}
\end{center}
\caption{Fifth step of $\Gamma$'s substitution system (time goes from bottom to top).\label{fig:gamma_steps}}
\end{figure}
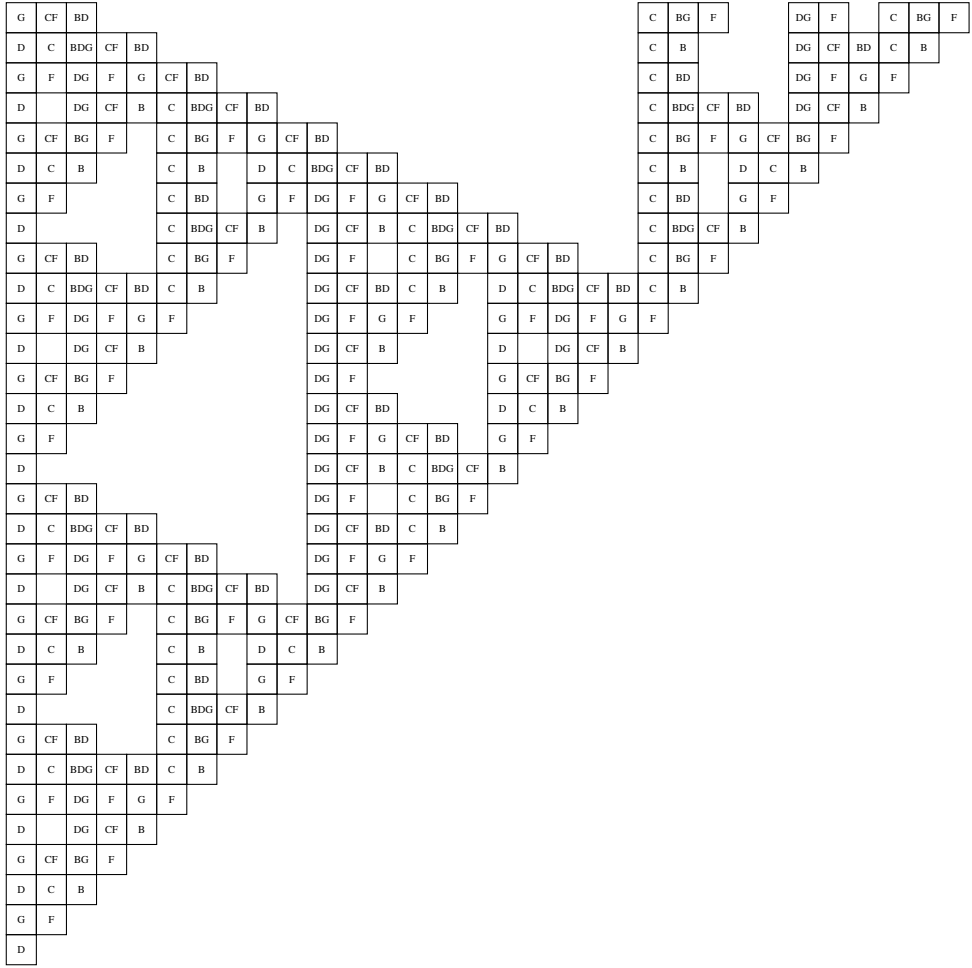

Whereas we have a theoretical number of $2^5=32$ different states in the substitution scheme, only 11 of them are accessible from the initial state, namely 0 plus the ones represented in Figure~\ref{fig:graphe}.  This graph has two strongly connected components, one composed of $BD$ alone, the other of the remaining vertices.  In particular, from any state of the substitution system that has been accessed from the initial state and that is neither $0$ nor $BD$, there is a path to $D$; therefore there must be a point of $X_2$ in the corresponding square.

\begin{figure}[htbp]
  \begin{center}
    \begin{tikzpicture}[node distance=2cm,>=stealth',bend angle=30,auto, circle,draw]
      \tikzstyle{nice}=[circle,thick,draw,fill=gray!10,minimum size=1cm]%
      \tikzstyle{mydblarrow}=[<->,very thick]
      \tikzstyle{myarrow}=[->,very thick]
      \node [nice] (D) {D};%
      \node [nice] (G) [right of=D] {G};%
      \node [nice] (CF) [right of=G] {CF};%
      \node [nice] (BDG) [right of=CF] {BDG};%
      \node [nice] (F) [below of=D] {F};%
      \node [nice] (C) [below right of=G] {C};%
      \node [nice] (B) [below of=C] {B};%
      \node [nice] (DG) [right of=B] {DG};%
      \node [nice] (BG) [left of=B] {BG};%
      \node [nice,fill=red!15] (BD) [below of=B] {BD};%
  
      \draw[mydblarrow] (D) edge (G);%
      \draw[mydblarrow] (C) edge (DG);%
      \draw[mydblarrow] (CF) edge (BDG);%
      \draw[myarrow] (G) edge (CF);%
      \draw[myarrow] (C) edge (F);%
      \draw[myarrow] (D) edge (F);%
      \draw[myarrow] (F) edge (BG);%
      \draw[myarrow] (G) edge (C);%
      \draw[myarrow] (B) edge [bend left] (G);%
      \draw[myarrow] (BDG) edge (C);%
      \draw[myarrow] (BDG) edge (B);%
      \draw[myarrow] (B) edge (F);%
      \draw[myarrow] (BG) edge (C);%
      \draw[myarrow] (BG) edge (BD);%
      \draw[myarrow] (C) edge (CF);%
      \draw[myarrow] (CF) edge [bend left] (BD);%
      \draw[myarrow] (D) edge [loop above] (D);%
      \draw[myarrow] (G) edge [loop above] (G);%
      \draw[myarrow] (CF) edge [loop above] (CF);%
      \draw[myarrow] (F) edge [loop left] (F);%
      \draw[myarrow] (BD) edge [loop left] (BD);%
    \end{tikzpicture}
  \end{center}
\caption{Transition graph of the substitution system: an arrow from state $s_1$ to state $s_2$ means that $s_2$ can be obtained after a finite number of iterations starting from $s_1$.\label{fig:graphe}}
\end{figure}

\subsubsection{A substitution system for $\Gamma^{-1}$}

We now have to perform the equivalent analysis for $\Gamma^{-1}$, which we will name $\Omega$, in order to avoid possible confusions with negative exponents.  The minimal polynomial of $\Omega$ is $X^3+(1+u^2)X^2+X+1$, so now the recurrence relation is

\begin{equation}
\forall x\in\Z\forall n,y\in\N\; \qquad y<3 \cdot 2^{n}\Longrightarrow
\label{eq:omega}
\Omega^{3\cdot 2^n+y}_x=\Omega^{2^{n+1}+y}_x+\Omega^{2^{n+1}+y}_{x-2^{n+1}}+\Omega^{2^n+y}_x+\Omega^y_x.
\end{equation}

We introduce $\beta$, which is to $\Omega$ what $\alpha$ was to $\Gamma$ in Section~\ref{sec:gamma}.

\begin{equation}
\Xi_{2x+s}^{2y+t}=\sum\limits_i \beta_{0}(x-i,y)\Xi_{2i+s}^t + \beta_{1}(x-i,y)\Xi_{2i+s}^{2+t}+\beta_{2}(x-i,y)\Xi_{2i+s}^{4+t}
\end{equation}

We then get the following decompositions.

\begin{equation}
\begin{array}{rcl}
\Xi_{2x}^{2y} &= &\sum\limits_i \beta_{0}(x-i,y)\Xi_{2i}^0 + \beta_{1}(x-i,y)\Xi_{2i}^{2}+\beta_{2}(x-i,y) \Xi_{2i}^4 \\
& = &\sum\limits_i \beta_{0}(x-i,y)\Xi_{2i}^0 + \beta_{1}(x-i,y)\Xi_{2i}^{2}+\beta_{2}(x-i,y) \pa{\Xi_{2i-4}^2 + \Xi_{2i-2}^1 + \Xi_{2i-2}^0 + \Xi_{2i}^0} \\
&=& \sum\limits_i \pa{\beta_{0}(x-i,y)+\beta_2 (x-1-i,y)+\beta_2 (x-i,y)}\Xi_{2i}^0 +\beta_{2}(x-1-i,y)\Xi_{2i}^{1}\\
&& + \pa{\beta_{1}(x-i,y)+\beta_2 (x-2-i,y)}\Xi_{2i}^{2} \\
\end{array}
\end{equation}

This is to be compared to this definition of $\beta_j$:
\begin{equation}
\Xi_{2x}^{2y}= \sum\limits_i\sum_{j=0}^2 \beta_{j}(2x-i,2y)\Xi_{i}^j
\end{equation}

Likewise, for $y\mapsto 2y+1$, we get

\begin{equation}
\begin{array}{rcl}
\Xi_{2x}^{2y+1} &=
&  \sum\limits_i \pa{ \beta_{1}(x-i,y)+\beta_{2}(x-2-i,y)}\Xi_{2i}^0 \\
& &  + [\beta_{0}(x-i,y)+\beta_{1}(x-i,y)\\
&& \hfill +\beta_{2}(x-2-i,y)+\beta_{2}(x-1-i,y)+\beta_{2}(x-i,y)]\Xi_{2i}^1 \\
&& + [\beta_{1}(x-1-i,y)+\beta_{1}(x-i,y)\\
&& \hfill +\beta_{2}(x-3-i,y)+\beta_{2}(x-2-i,y)+\beta_{2}(x-1-i,y)]\Xi_{2i}^2 \\
\end{array}
\end{equation}

The minimal grouping is now\ldots

$$\begin{array}{|cccccc|}
  \hline
  \beta_{2}(x-5,y)&\beta_{2}(x-4,y)&\beta_{2}(x-3,y)&\beta_{2}(x-2,y)&\beta_{2}(x-1,y)&\beta_{2}(x,y) \\
&&\beta_{1}(x-3,y)&\beta_{1}(x-2,y)&\beta_{1}(x-1,y)&\beta_{1}(x,y) \\
&&& & \beta_{0}(x-1,y)& \beta_{0}(x,y) \\
  \hline
\end{array}$$

\ldots and the corresponding substitution scheme is given by

\begin{center}
\tiny
\scalebox{0.9}
{$\begin{array}{c}
 \begin{array}{|c|}
    \hline
    \begin{array}{cccccc}
      a & b & c & d&e&f \\
      &&g& h & i & j \\
      &&&& k & l \\
    \end{array}\\
    \hline
  \end{array}
\\
  \downarrow
\\
  \begin{array}{|cccccc|cccccc|}
    \hline
    
      0& a+b+c+g+h & 0 & b+c+d+h+i & 0 & c+d+e+i+j &
      		a+b+c+g+h & 0 & b+c+d+h+i & 0 & c+d+e+i+j & 0 \\
       && 0&  c+d+e+i+k  &  0  & d+e+f+j+l &  
		&& c+d+e+i+k & 0 & d+e+f+j+l & 0  \\
      &&&& 0 & d+j &
		&&&&d+j&0  \\ 

    \hline

      0 & b+h & 0 & c+i & 0 & d+j &
		b+h & 0 & c+i & 0 & d+j & 0 \\
      & & 0 &d &0&e &
		&&d&0&e&0  \\
      & & & & 0 & e+f+l &
      		&&&& e+f+l & 0\\
    \hline
  \end{array}
\end{array}
$}
\end{center}

The initial state is $\begin{array}{c|c|c|c|c|c}\hline \cdots&0&L&K&0&\cdots\\ \hline\end{array}$, and 
$\left(\Gamma^{-1}\right)_x^y$ is given by 

$$\pa{\begin{array}{ccc}l+h+f+d+b & i+e+c &j+d \\i+e+c &l+j+f+d & e\\j+d &e &l+f \\ \end{array}}.$$

$A$ being equivalent to $G$, $B$ to $H$ and $F$ to $L$, we get the simpler

\begin{center}
\scalebox{1}
{$\begin{array}{c}
 \begin{array}{|c|}
    \hline
    \begin{array}{cccc}
      c & d&e& \\
      g& h & i & j \\
      && k & l \\
    \end{array}\\
    \hline
  \end{array}
\\
  \downarrow
\\
  \begin{array}{|cccc|cccc|}
    \hline
    
      0 & c+d+h+i & 0 &  &
      	  c+d+h+i & 0 & c+d+e+i+j& \\
        0&  d+e+g+h+i+k  &  0  & d+e+j+l &  
		d+e+g+h+i+k & 0 & d+e+j+l & 0  \\
      && 0 & c+e+i &
		&&d+j&0  \\ 

    \hline

       0 & c+i & 0 & &
	 	 c+i & 0 & d+j & \\
       0 &d+h &0&e &
		d+h&0&e&0  \\
       & & 0 & d+e+j+l &
      		&& e+l & 0\\
    \hline
  \end{array}
\end{array}
$}
\end{center}

This first simplification makes $G$ equivalent to $K$, $C$ to $IK$ and $J$ to $DH$, so we finally get

\begin{center}
\scalebox{1}
{$\begin{array}{c}
 \begin{array}{|c|}
    \hline
    \begin{array}{cc}
        d&e \\
       h & i   \\
       k & l \\
    \end{array}\\
    \hline
  \end{array}
\\
  \downarrow
\\
  \begin{array}{|cc|cc|}
    \hline
    
      e+h+i+l & 0 &
      	  0 & d+e+i \\
         h+i+k+l  &  0 &
		 0 & e+h+l \\
      0 & e+i &
		d+k&0  \\ 

    \hline

       e+i & 0   &
	 	  0 & d  \\
       d+e+h &0&
		 0&e+i  \\
       0 & d+e+l &
      		d+e+h+i+l & 0\\
    \hline
  \end{array}
\end{array}
$}
\end{center}

\ldots which results after five steps in Figure~\ref{fig:omega_steps}.

\renewcommand\metsmoicaici[3]{\draw (#2,-#1) node[draw,minimum size=.35cm] () {#3};}

\begin{figure}[htbp]
\begin{center}
  \fontsize{3}{3}\selectfont
  \begin{tikzpicture}[xscale=.35,yscale=.35]
\metsmoicaici{63}{31}{K}\metsmoicaici{62}{31}{H}\metsmoicaici{61}{31}{I}\metsmoicaici{60}{31}{DH}\metsmoicaici{59}{31}{EI}\metsmoicaici{58}{31}{DHL}\metsmoicaici{57}{31}{EIK}\metsmoicaici{56}{31}{DH}\metsmoicaici{55}{31}{K}\metsmoicaici{54}{31}{H}\metsmoicaici{53}{31}{EK}\metsmoicaici{49}{31}{EIK}\metsmoicaici{48}{31}{DH}\metsmoicaici{47}{31}{K}\metsmoicaici{46}{31}{H}\metsmoicaici{45}{31}{I}\metsmoicaici{44}{31}{DH}\metsmoicaici{43}{31}{EI}\metsmoicaici{42}{31}{DHL}\metsmoicaici{33}{31}{EIK}\metsmoicaici{32}{31}{DH}\metsmoicaici{31}{31}{K}\metsmoicaici{30}{31}{H}\metsmoicaici{29}{31}{I}\metsmoicaici{28}{31}{DH}\metsmoicaici{27}{31}{EI}\metsmoicaici{26}{31}{DHL}\metsmoicaici{25}{31}{EIK}\metsmoicaici{24}{31}{DH}\metsmoicaici{23}{31}{K}\metsmoicaici{22}{31}{H}\metsmoicaici{21}{31}{EK}\metsmoicaici{1}{31}{EIK}\metsmoicaici{0}{31}{DH}
\metsmoicaici{61}{30}{K}\metsmoicaici{60}{30}{H}\metsmoicaici{59}{30}{IK}\metsmoicaici{58}{30}{D}\metsmoicaici{57}{30}{E}\metsmoicaici{56}{30}{L}\metsmoicaici{54}{30}{HL}\metsmoicaici{53}{30}{EK}\metsmoicaici{50}{30}{HL}\metsmoicaici{49}{30}{E}\metsmoicaici{48}{30}{L}\metsmoicaici{45}{30}{K}\metsmoicaici{44}{30}{H}\metsmoicaici{43}{30}{IK}\metsmoicaici{42}{30}{DHL}\metsmoicaici{34}{30}{HL}\metsmoicaici{33}{30}{E}\metsmoicaici{32}{30}{L}\metsmoicaici{29}{30}{K}\metsmoicaici{28}{30}{H}\metsmoicaici{27}{30}{IK}\metsmoicaici{26}{30}{D}\metsmoicaici{25}{30}{E}\metsmoicaici{24}{30}{L}\metsmoicaici{22}{30}{HL}\metsmoicaici{21}{30}{EK}\metsmoicaici{2}{30}{HL}\metsmoicaici{1}{30}{E}\metsmoicaici{0}{30}{L}
\metsmoicaici{59}{29}{K}\metsmoicaici{58}{29}{H}\metsmoicaici{57}{29}{I}\metsmoicaici{56}{29}{DH}\metsmoicaici{55}{29}{EIK}\metsmoicaici{54}{29}{DL}\metsmoicaici{53}{29}{EK}\metsmoicaici{51}{29}{EIK}\metsmoicaici{50}{29}{DL}\metsmoicaici{49}{29}{I}\metsmoicaici{48}{29}{DH}\metsmoicaici{43}{29}{EI}\metsmoicaici{42}{29}{DHL}\metsmoicaici{35}{29}{EIK}\metsmoicaici{34}{29}{DL}\metsmoicaici{33}{29}{I}\metsmoicaici{32}{29}{DH}\metsmoicaici{27}{29}{K}\metsmoicaici{26}{29}{H}\metsmoicaici{25}{29}{I}\metsmoicaici{24}{29}{DH}\metsmoicaici{23}{29}{EIK}\metsmoicaici{22}{29}{DL}\metsmoicaici{21}{29}{EK}\metsmoicaici{3}{29}{EIK}\metsmoicaici{2}{29}{DL}\metsmoicaici{1}{29}{I}\metsmoicaici{0}{29}{DH}
\metsmoicaici{57}{28}{K}\metsmoicaici{56}{28}{H}\metsmoicaici{55}{28}{IK}\metsmoicaici{54}{28}{D}\metsmoicaici{53}{28}{EK}\metsmoicaici{52}{28}{HL}\metsmoicaici{51}{28}{IK}\metsmoicaici{50}{28}{DHL}\metsmoicaici{49}{28}{K}\metsmoicaici{48}{28}{L}\metsmoicaici{44}{28}{HL}\metsmoicaici{43}{28}{IK}\metsmoicaici{42}{28}{DHL}\metsmoicaici{36}{28}{HL}\metsmoicaici{35}{28}{IK}\metsmoicaici{34}{28}{DHL}\metsmoicaici{33}{28}{K}\metsmoicaici{32}{28}{L}\metsmoicaici{25}{28}{K}\metsmoicaici{24}{28}{H}\metsmoicaici{23}{28}{IK}\metsmoicaici{22}{28}{D}\metsmoicaici{21}{28}{EK}\metsmoicaici{4}{28}{HL}\metsmoicaici{3}{28}{IK}\metsmoicaici{2}{28}{DHL}\metsmoicaici{1}{28}{K}\metsmoicaici{0}{28}{L}
\metsmoicaici{55}{27}{K}\metsmoicaici{54}{27}{H}\metsmoicaici{53}{27}{I}\metsmoicaici{52}{27}{DH}\metsmoicaici{51}{27}{EI}\metsmoicaici{50}{27}{DHL}\metsmoicaici{49}{27}{EIK}\metsmoicaici{48}{27}{DH}\metsmoicaici{45}{27}{EIK}\metsmoicaici{44}{27}{DH}\metsmoicaici{43}{27}{EI}\metsmoicaici{42}{27}{DHL}\metsmoicaici{37}{27}{EIK}\metsmoicaici{36}{27}{DH}\metsmoicaici{35}{27}{EI}\metsmoicaici{34}{27}{DHL}\metsmoicaici{33}{27}{EIK}\metsmoicaici{32}{27}{DH}\metsmoicaici{23}{27}{K}\metsmoicaici{22}{27}{H}\metsmoicaici{21}{27}{EK}\metsmoicaici{5}{27}{EIK}\metsmoicaici{4}{27}{DH}\metsmoicaici{3}{27}{EI}\metsmoicaici{2}{27}{DHL}\metsmoicaici{1}{27}{EIK}\metsmoicaici{0}{27}{DH}
\metsmoicaici{53}{26}{K}\metsmoicaici{52}{26}{H}\metsmoicaici{51}{26}{IK}\metsmoicaici{50}{26}{D}\metsmoicaici{49}{26}{E}\metsmoicaici{48}{26}{L}\metsmoicaici{46}{26}{HL}\metsmoicaici{45}{26}{E}\metsmoicaici{44}{26}{H}\metsmoicaici{43}{26}{IK}\metsmoicaici{42}{26}{DHL}\metsmoicaici{38}{26}{HL}\metsmoicaici{37}{26}{E}\metsmoicaici{36}{26}{H}\metsmoicaici{35}{26}{IK}\metsmoicaici{34}{26}{D}\metsmoicaici{33}{26}{E}\metsmoicaici{32}{26}{L}\metsmoicaici{22}{26}{HL}\metsmoicaici{21}{26}{EK}\metsmoicaici{6}{26}{HL}\metsmoicaici{5}{26}{E}\metsmoicaici{4}{26}{H}\metsmoicaici{3}{26}{IK}\metsmoicaici{2}{26}{D}\metsmoicaici{1}{26}{E}\metsmoicaici{0}{26}{L}
\metsmoicaici{51}{25}{K}\metsmoicaici{50}{25}{H}\metsmoicaici{49}{25}{I}\metsmoicaici{48}{25}{DH}\metsmoicaici{47}{25}{EIK}\metsmoicaici{46}{25}{DL}\metsmoicaici{45}{25}{EK}\metsmoicaici{43}{25}{EI}\metsmoicaici{42}{25}{DHL}\metsmoicaici{39}{25}{EIK}\metsmoicaici{38}{25}{DL}\metsmoicaici{37}{25}{EK}\metsmoicaici{35}{25}{K}\metsmoicaici{34}{25}{H}\metsmoicaici{33}{25}{I}\metsmoicaici{32}{25}{DH}\metsmoicaici{23}{25}{EIK}\metsmoicaici{22}{25}{DL}\metsmoicaici{21}{25}{EK}\metsmoicaici{7}{25}{EIK}\metsmoicaici{6}{25}{DL}\metsmoicaici{5}{25}{EK}\metsmoicaici{3}{25}{K}\metsmoicaici{2}{25}{H}\metsmoicaici{1}{25}{I}\metsmoicaici{0}{25}{DH}
\metsmoicaici{49}{24}{K}\metsmoicaici{48}{24}{H}\metsmoicaici{47}{24}{IK}\metsmoicaici{46}{24}{D}\metsmoicaici{45}{24}{EK}\metsmoicaici{44}{24}{HL}\metsmoicaici{43}{24}{IK}\metsmoicaici{42}{24}{DHL}\metsmoicaici{40}{24}{HL}\metsmoicaici{39}{24}{IK}\metsmoicaici{38}{24}{D}\metsmoicaici{37}{24}{EK}\metsmoicaici{33}{24}{K}\metsmoicaici{32}{24}{L}\metsmoicaici{24}{24}{HL}\metsmoicaici{23}{24}{IK}\metsmoicaici{22}{24}{D}\metsmoicaici{21}{24}{EK}\metsmoicaici{8}{24}{HL}\metsmoicaici{7}{24}{IK}\metsmoicaici{6}{24}{D}\metsmoicaici{5}{24}{EK}\metsmoicaici{1}{24}{K}\metsmoicaici{0}{24}{L}
\metsmoicaici{47}{23}{K}\metsmoicaici{46}{23}{H}\metsmoicaici{45}{23}{I}\metsmoicaici{44}{23}{DH}\metsmoicaici{43}{23}{EI}\metsmoicaici{42}{23}{DHL}\metsmoicaici{41}{23}{EIK}\metsmoicaici{40}{23}{DH}\metsmoicaici{39}{23}{K}\metsmoicaici{38}{23}{H}\metsmoicaici{37}{23}{EK}\metsmoicaici{33}{23}{EIK}\metsmoicaici{32}{23}{DH}\metsmoicaici{25}{23}{EIK}\metsmoicaici{24}{23}{DH}\metsmoicaici{23}{23}{K}\metsmoicaici{22}{23}{H}\metsmoicaici{21}{23}{EK}\metsmoicaici{9}{23}{EIK}\metsmoicaici{8}{23}{DH}\metsmoicaici{7}{23}{K}\metsmoicaici{6}{23}{H}\metsmoicaici{5}{23}{EK}\metsmoicaici{1}{23}{EIK}\metsmoicaici{0}{23}{DH}
\metsmoicaici{45}{22}{K}\metsmoicaici{44}{22}{H}\metsmoicaici{43}{22}{IK}\metsmoicaici{42}{22}{D}\metsmoicaici{41}{22}{E}\metsmoicaici{40}{22}{L}\metsmoicaici{38}{22}{HL}\metsmoicaici{37}{22}{EK}\metsmoicaici{34}{22}{HL}\metsmoicaici{33}{22}{E}\metsmoicaici{32}{22}{L}\metsmoicaici{26}{22}{HL}\metsmoicaici{25}{22}{E}\metsmoicaici{24}{22}{L}\metsmoicaici{22}{22}{HL}\metsmoicaici{21}{22}{EK}\metsmoicaici{10}{22}{HL}\metsmoicaici{9}{22}{E}\metsmoicaici{8}{22}{L}\metsmoicaici{6}{22}{HL}\metsmoicaici{5}{22}{EK}\metsmoicaici{2}{22}{HL}\metsmoicaici{1}{22}{E}\metsmoicaici{0}{22}{L}
\metsmoicaici{43}{21}{K}\metsmoicaici{42}{21}{H}\metsmoicaici{41}{21}{I}\metsmoicaici{40}{21}{DH}\metsmoicaici{39}{21}{EIK}\metsmoicaici{38}{21}{DL}\metsmoicaici{37}{21}{EK}\metsmoicaici{35}{21}{EIK}\metsmoicaici{34}{21}{DL}\metsmoicaici{33}{21}{I}\metsmoicaici{32}{21}{DH}\metsmoicaici{27}{21}{EIK}\metsmoicaici{26}{21}{DL}\metsmoicaici{25}{21}{I}\metsmoicaici{24}{21}{DH}\metsmoicaici{23}{21}{EIK}\metsmoicaici{22}{21}{DL}\metsmoicaici{21}{21}{EK}\metsmoicaici{11}{21}{EIK}\metsmoicaici{10}{21}{DL}\metsmoicaici{9}{21}{I}\metsmoicaici{8}{21}{DH}\metsmoicaici{7}{21}{EIK}\metsmoicaici{6}{21}{DL}\metsmoicaici{5}{21}{EK}\metsmoicaici{3}{21}{EIK}\metsmoicaici{2}{21}{DL}\metsmoicaici{1}{21}{I}\metsmoicaici{0}{21}{DH}
\metsmoicaici{41}{20}{K}\metsmoicaici{40}{20}{H}\metsmoicaici{39}{20}{IK}\metsmoicaici{38}{20}{D}\metsmoicaici{37}{20}{EK}\metsmoicaici{36}{20}{HL}\metsmoicaici{35}{20}{IK}\metsmoicaici{34}{20}{DHL}\metsmoicaici{33}{20}{K}\metsmoicaici{32}{20}{L}\metsmoicaici{28}{20}{HL}\metsmoicaici{27}{20}{IK}\metsmoicaici{26}{20}{DHL}\metsmoicaici{25}{20}{K}\metsmoicaici{24}{20}{H}\metsmoicaici{23}{20}{IK}\metsmoicaici{22}{20}{D}\metsmoicaici{21}{20}{EK}\metsmoicaici{12}{20}{HL}\metsmoicaici{11}{20}{IK}\metsmoicaici{10}{20}{DHL}\metsmoicaici{9}{20}{K}\metsmoicaici{8}{20}{H}\metsmoicaici{7}{20}{IK}\metsmoicaici{6}{20}{D}\metsmoicaici{5}{20}{EK}\metsmoicaici{4}{20}{HL}\metsmoicaici{3}{20}{IK}\metsmoicaici{2}{20}{DHL}\metsmoicaici{1}{20}{K}\metsmoicaici{0}{20}{L}
\metsmoicaici{39}{19}{K}\metsmoicaici{38}{19}{H}\metsmoicaici{37}{19}{I}\metsmoicaici{36}{19}{DH}\metsmoicaici{35}{19}{EI}\metsmoicaici{34}{19}{DHL}\metsmoicaici{33}{19}{EIK}\metsmoicaici{32}{19}{DH}\metsmoicaici{29}{19}{EIK}\metsmoicaici{28}{19}{DH}\metsmoicaici{27}{19}{EI}\metsmoicaici{26}{19}{DHL}\metsmoicaici{23}{19}{K}\metsmoicaici{22}{19}{H}\metsmoicaici{21}{19}{EK}\metsmoicaici{13}{19}{EIK}\metsmoicaici{12}{19}{DH}\metsmoicaici{11}{19}{EI}\metsmoicaici{10}{19}{DHL}\metsmoicaici{7}{19}{K}\metsmoicaici{6}{19}{H}\metsmoicaici{5}{19}{I}\metsmoicaici{4}{19}{DH}\metsmoicaici{3}{19}{EI}\metsmoicaici{2}{19}{DHL}\metsmoicaici{1}{19}{EIK}\metsmoicaici{0}{19}{DH}
\metsmoicaici{37}{18}{K}\metsmoicaici{36}{18}{H}\metsmoicaici{35}{18}{IK}\metsmoicaici{34}{18}{D}\metsmoicaici{33}{18}{E}\metsmoicaici{32}{18}{L}\metsmoicaici{30}{18}{HL}\metsmoicaici{29}{18}{E}\metsmoicaici{28}{18}{H}\metsmoicaici{27}{18}{IK}\metsmoicaici{26}{18}{DHL}\metsmoicaici{22}{18}{HL}\metsmoicaici{21}{18}{EK}\metsmoicaici{14}{18}{HL}\metsmoicaici{13}{18}{E}\metsmoicaici{12}{18}{H}\metsmoicaici{11}{18}{IK}\metsmoicaici{10}{18}{DHL}\metsmoicaici{5}{18}{K}\metsmoicaici{4}{18}{H}\metsmoicaici{3}{18}{IK}\metsmoicaici{2}{18}{D}\metsmoicaici{1}{18}{E}\metsmoicaici{0}{18}{L}
\metsmoicaici{35}{17}{K}\metsmoicaici{34}{17}{H}\metsmoicaici{33}{17}{I}\metsmoicaici{32}{17}{DH}\metsmoicaici{31}{17}{EIK}\metsmoicaici{30}{17}{DL}\metsmoicaici{29}{17}{EK}\metsmoicaici{27}{17}{EI}\metsmoicaici{26}{17}{DHL}\metsmoicaici{23}{17}{EIK}\metsmoicaici{22}{17}{DL}\metsmoicaici{21}{17}{EK}\metsmoicaici{15}{17}{EIK}\metsmoicaici{14}{17}{DL}\metsmoicaici{13}{17}{EK}\metsmoicaici{11}{17}{EI}\metsmoicaici{10}{17}{DHL}\metsmoicaici{3}{17}{K}\metsmoicaici{2}{17}{H}\metsmoicaici{1}{17}{I}\metsmoicaici{0}{17}{DH}
\metsmoicaici{33}{16}{K}\metsmoicaici{32}{16}{H}\metsmoicaici{31}{16}{IK}\metsmoicaici{30}{16}{D}\metsmoicaici{29}{16}{EK}\metsmoicaici{28}{16}{HL}\metsmoicaici{27}{16}{IK}\metsmoicaici{26}{16}{DHL}\metsmoicaici{24}{16}{HL}\metsmoicaici{23}{16}{IK}\metsmoicaici{22}{16}{D}\metsmoicaici{21}{16}{EK}\metsmoicaici{16}{16}{HL}\metsmoicaici{15}{16}{IK}\metsmoicaici{14}{16}{D}\metsmoicaici{13}{16}{EK}\metsmoicaici{12}{16}{HL}\metsmoicaici{11}{16}{IK}\metsmoicaici{10}{16}{DHL}\metsmoicaici{1}{16}{K}\metsmoicaici{0}{16}{L}
\metsmoicaici{31}{15}{K}\metsmoicaici{30}{15}{H}\metsmoicaici{29}{15}{I}\metsmoicaici{28}{15}{DH}\metsmoicaici{27}{15}{EI}\metsmoicaici{26}{15}{DHL}\metsmoicaici{25}{15}{EIK}\metsmoicaici{24}{15}{DH}\metsmoicaici{23}{15}{K}\metsmoicaici{22}{15}{H}\metsmoicaici{21}{15}{EK}\metsmoicaici{17}{15}{EIK}\metsmoicaici{16}{15}{DH}\metsmoicaici{15}{15}{K}\metsmoicaici{14}{15}{H}\metsmoicaici{13}{15}{I}\metsmoicaici{12}{15}{DH}\metsmoicaici{11}{15}{EI}\metsmoicaici{10}{15}{DHL}\metsmoicaici{1}{15}{EIK}\metsmoicaici{0}{15}{DH}
\metsmoicaici{29}{14}{K}\metsmoicaici{28}{14}{H}\metsmoicaici{27}{14}{IK}\metsmoicaici{26}{14}{D}\metsmoicaici{25}{14}{E}\metsmoicaici{24}{14}{L}\metsmoicaici{22}{14}{HL}\metsmoicaici{21}{14}{EK}\metsmoicaici{18}{14}{HL}\metsmoicaici{17}{14}{E}\metsmoicaici{16}{14}{L}\metsmoicaici{13}{14}{K}\metsmoicaici{12}{14}{H}\metsmoicaici{11}{14}{IK}\metsmoicaici{10}{14}{DHL}\metsmoicaici{2}{14}{HL}\metsmoicaici{1}{14}{E}\metsmoicaici{0}{14}{L}
\metsmoicaici{27}{13}{K}\metsmoicaici{26}{13}{H}\metsmoicaici{25}{13}{I}\metsmoicaici{24}{13}{DH}\metsmoicaici{23}{13}{EIK}\metsmoicaici{22}{13}{DL}\metsmoicaici{21}{13}{EK}\metsmoicaici{19}{13}{EIK}\metsmoicaici{18}{13}{DL}\metsmoicaici{17}{13}{I}\metsmoicaici{16}{13}{DH}\metsmoicaici{11}{13}{EI}\metsmoicaici{10}{13}{DHL}\metsmoicaici{3}{13}{EIK}\metsmoicaici{2}{13}{DL}\metsmoicaici{1}{13}{I}\metsmoicaici{0}{13}{DH}
\metsmoicaici{25}{12}{K}\metsmoicaici{24}{12}{H}\metsmoicaici{23}{12}{IK}\metsmoicaici{22}{12}{D}\metsmoicaici{21}{12}{EK}\metsmoicaici{20}{12}{HL}\metsmoicaici{19}{12}{IK}\metsmoicaici{18}{12}{DHL}\metsmoicaici{17}{12}{K}\metsmoicaici{16}{12}{L}\metsmoicaici{12}{12}{HL}\metsmoicaici{11}{12}{IK}\metsmoicaici{10}{12}{DHL}\metsmoicaici{4}{12}{HL}\metsmoicaici{3}{12}{IK}\metsmoicaici{2}{12}{DHL}\metsmoicaici{1}{12}{K}\metsmoicaici{0}{12}{L}
\metsmoicaici{23}{11}{K}\metsmoicaici{22}{11}{H}\metsmoicaici{21}{11}{I}\metsmoicaici{20}{11}{DH}\metsmoicaici{19}{11}{EI}\metsmoicaici{18}{11}{DHL}\metsmoicaici{17}{11}{EIK}\metsmoicaici{16}{11}{DH}\metsmoicaici{13}{11}{EIK}\metsmoicaici{12}{11}{DH}\metsmoicaici{11}{11}{EI}\metsmoicaici{10}{11}{DHL}\metsmoicaici{5}{11}{EIK}\metsmoicaici{4}{11}{DH}\metsmoicaici{3}{11}{EI}\metsmoicaici{2}{11}{DHL}\metsmoicaici{1}{11}{EIK}\metsmoicaici{0}{11}{DH}
\metsmoicaici{21}{10}{K}\metsmoicaici{20}{10}{H}\metsmoicaici{19}{10}{IK}\metsmoicaici{18}{10}{D}\metsmoicaici{17}{10}{E}\metsmoicaici{16}{10}{L}\metsmoicaici{14}{10}{HL}\metsmoicaici{13}{10}{E}\metsmoicaici{12}{10}{H}\metsmoicaici{11}{10}{IK}\metsmoicaici{10}{10}{DHL}\metsmoicaici{6}{10}{HL}\metsmoicaici{5}{10}{E}\metsmoicaici{4}{10}{H}\metsmoicaici{3}{10}{IK}\metsmoicaici{2}{10}{D}\metsmoicaici{1}{10}{E}\metsmoicaici{0}{10}{L}
\metsmoicaici{19}{9}{K}\metsmoicaici{18}{9}{H}\metsmoicaici{17}{9}{I}\metsmoicaici{16}{9}{DH}\metsmoicaici{15}{9}{EIK}\metsmoicaici{14}{9}{DL}\metsmoicaici{13}{9}{EK}\metsmoicaici{11}{9}{EI}\metsmoicaici{10}{9}{DHL}\metsmoicaici{7}{9}{EIK}\metsmoicaici{6}{9}{DL}\metsmoicaici{5}{9}{EK}\metsmoicaici{3}{9}{K}\metsmoicaici{2}{9}{H}\metsmoicaici{1}{9}{I}\metsmoicaici{0}{9}{DH}
\metsmoicaici{17}{8}{K}\metsmoicaici{16}{8}{H}\metsmoicaici{15}{8}{IK}\metsmoicaici{14}{8}{D}\metsmoicaici{13}{8}{EK}\metsmoicaici{12}{8}{HL}\metsmoicaici{11}{8}{IK}\metsmoicaici{10}{8}{DHL}\metsmoicaici{8}{8}{HL}\metsmoicaici{7}{8}{IK}\metsmoicaici{6}{8}{D}\metsmoicaici{5}{8}{EK}\metsmoicaici{1}{8}{K}\metsmoicaici{0}{8}{L}
\metsmoicaici{15}{7}{K}\metsmoicaici{14}{7}{H}\metsmoicaici{13}{7}{I}\metsmoicaici{12}{7}{DH}\metsmoicaici{11}{7}{EI}\metsmoicaici{10}{7}{DHL}\metsmoicaici{9}{7}{EIK}\metsmoicaici{8}{7}{DH}\metsmoicaici{7}{7}{K}\metsmoicaici{6}{7}{H}\metsmoicaici{5}{7}{EK}\metsmoicaici{1}{7}{EIK}\metsmoicaici{0}{7}{DH}
\metsmoicaici{13}{6}{K}\metsmoicaici{12}{6}{H}\metsmoicaici{11}{6}{IK}\metsmoicaici{10}{6}{D}\metsmoicaici{9}{6}{E}\metsmoicaici{8}{6}{L}\metsmoicaici{6}{6}{HL}\metsmoicaici{5}{6}{EK}\metsmoicaici{2}{6}{HL}\metsmoicaici{1}{6}{E}\metsmoicaici{0}{6}{L}
\metsmoicaici{11}{5}{K}\metsmoicaici{10}{5}{H}\metsmoicaici{9}{5}{I}\metsmoicaici{8}{5}{DH}\metsmoicaici{7}{5}{EIK}\metsmoicaici{6}{5}{DL}\metsmoicaici{5}{5}{EK}\metsmoicaici{3}{5}{EIK}\metsmoicaici{2}{5}{DL}\metsmoicaici{1}{5}{I}\metsmoicaici{0}{5}{DH}
\metsmoicaici{9}{4}{K}\metsmoicaici{8}{4}{H}\metsmoicaici{7}{4}{IK}\metsmoicaici{6}{4}{D}\metsmoicaici{5}{4}{EK}\metsmoicaici{4}{4}{HL}\metsmoicaici{3}{4}{IK}\metsmoicaici{2}{4}{DHL}\metsmoicaici{1}{4}{K}\metsmoicaici{0}{4}{L}
\metsmoicaici{7}{3}{K}\metsmoicaici{6}{3}{H}\metsmoicaici{5}{3}{I}\metsmoicaici{4}{3}{DH}\metsmoicaici{3}{3}{EI}\metsmoicaici{2}{3}{DHL}\metsmoicaici{1}{3}{EIK}\metsmoicaici{0}{3}{DH}
\metsmoicaici{5}{2}{K}\metsmoicaici{4}{2}{H}\metsmoicaici{3}{2}{IK}\metsmoicaici{2}{2}{D}\metsmoicaici{1}{2}{E}\metsmoicaici{0}{2}{L}
\metsmoicaici{3}{1}{K}\metsmoicaici{2}{1}{H}\metsmoicaici{1}{1}{I}\metsmoicaici{0}{1}{DH}
\metsmoicaici{1}{0}{K}\metsmoicaici{0}{0}{L}

  \end{tikzpicture}
\end{center}
\caption{Fifth step of $\Gamma^{-1}$'s substitution system (time goes from left to right).\label{fig:omega_steps}}
\end{figure}

\subsubsection{Final arguments}
\label{sub:final}

It now remains to be proven that Figure~\ref{fig:gamma} does represent $X_2$ for $\Gamma$ and $\Gamma^{-1}$.  As such, this does not mean much; actually, we need to prove a few features of $X_2$ that would suffice in order to conclude that $\Gamma$ does not simulate $\Gamma^{-1}$.  Namely, we want to justify this series of assertions:

\begin{itemize}
\item[(i)] $\overline{X_2(\Gamma)}$ contains the (half-)lines $\R_+(0,1)$ and $\R_+(1,1)$.
\item[(ii)] $\overline{X_2(\Gamma)}$ contains the segment $\left[(0,1);(\frac{2}{3},\frac{2}{3})\right]$.
\item[(iii)] No point of $X_2(\Gamma^{-1})$ lies in the interior of the triangle with vertices $(0,1)$, $(\frac{2}{3},\frac{1}{3})$ and $(\frac{2}{3},\frac{4}{3})$. 
\end{itemize}

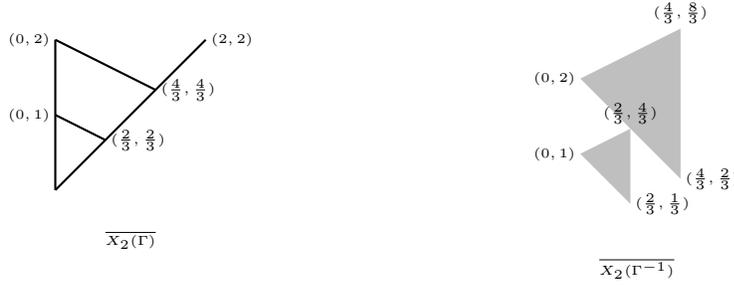
\begin{figure}
  \centering
  \begin{minipage}{.45\linewidth}
    \begin{center}
      \tiny
      \begin{tikzpicture}
        \draw[thick] (0,0)--(0,2);
        \draw[thick] (0,0)--(2,2) node [right] {$(2,2)$};
        \draw[thick] (0,1) node [left] {$(0,1)$} --(.666,.666) node [right] {$(\frac{2}{3},\frac{2}{3})$} ;
        \draw[thick] (0,2) node [left] {$(0,2)$} --(1.333,1.333) node [right] {$(\frac{4}{3},\frac{4}{3})$} ;
      \end{tikzpicture}
      \vskip .5cm
      $\overline{X_2(\Gamma)}$
    \end{center}
  \end{minipage}
  \begin{minipage}{.45\linewidth}
    \begin{center}
      \tiny
      \begin{tikzpicture}
        \fill[fill=gray!50!white] (0,2) node [left] {$(0,2)$} -- 
        (1.333,0.666) node [right] {$(\frac{4}{3},\frac{2}{3})$} -- 
        (1.333,2.666) node [above] {$(\frac{4}{3},\frac{8}{3})$} -- cycle;
        \fill[fill=gray!50!white] (0,1) node [left] {$(0,1)$} -- 
        (0.666,0.333) node [right] {$(\frac{2}{3},\frac{1}{3})$} -- 
        (0.666,1.333) node [above] {$(\frac{2}{3},\frac{4}{3})$} -- cycle;
      \end{tikzpicture}
      \vskip .5cm
      $\overline{X_2(\Gamma^{-1})}$
    \end{center}
  \end{minipage}
  \caption{Partial knowledge about $\overline{X_2(\Gamma)}$ and $\overline{X_2(\Gamma^{-1})}$. Points in black are known to belong to the set while points in gray are known not to belong to the set. Remember that $X_2$ is invariant by homothetic transformations of center $(0,0)$ and factor $2^i$ (with $i$ any integer).}
  \label{fig:assertions}
\end{figure}

This is enough to conclude, because (iii) implies that the only possible half-lines starting at the origin and included in $\overline{X_2(\Gamma^{-1})}$ are the vertical axis and that of slope $\frac{1}{2}$; and the segments joining these lines, if they exist, must have slope $-1$. Therefore it is impossible to send $X_2(\Gamma)$ into $X_2(\Gamma^{-1})$ by a $\pi_{\alpha,\beta,\gamma}$ transformation (see figure~\ref{fig:assertions}) and Theorem~\ref{thm:simu} concludes. Since $X_2({\tilde \Gamma})$ is just the symmetric of $X_2(\Gamma^{-1})$ with respect to the vertical axis passing through $(0,0)$, the same reasoning with Theorem~\ref{thm:simu} shows that $\Gamma$ cannot simulate ${\tilde \Gamma}$.

\begin{prop}
  $\Gamma$ can neither $\simu$-simulate its inverse $\Gamma^{-1}$ nor its dual ${\tilde \Gamma}$.
\end{prop}

We now prove the three assertions above successively using the substitution systems derived earlier.

\begin{ppty}
  $\overline{X_2(\Gamma)}$ contains the (half-)lines $\R_+(0,1)$ and $\R_+(1,1)$.
\end{ppty}
\begin{proof}
  First, by looking at the images of $D$, $G$, $B$ and $F$ by the substitution system of $\Gamma$, we prove by recurrence that:
  \begin{itemize}
  \item $\Upsilon(0,n)=D$ if $n$ is even and $G$ else, and
  \item $\Upsilon(n,n)=B$ if $n\geq 1$ is even and $F$ else,
  \end{itemize}
  where $\Upsilon$ is the fixed-point of the substitution.  We deduce from the former observation that every letter in the substitution system, except for $BD$, contains a point in $X_2$, that the (half-)lines $\R_+(0,1)$ and $\R_+(1,1)$ are in $\overline{X_2(\Gamma)}$.
\end{proof}

\begin{ppty}
  $\overline{X_2(\Gamma)}$ contains the segment $\left[(0,1);(\frac{2}{3},\frac{2}{3})\right]$.
\end{ppty}
\begin{proof}
The substitution system of $\Gamma$ is such that:
\begin{align*}
  BD &\rightarrow \begin{array}{|c|c|}\hline 0 & 0 \\ \hline BD & 0 \\ \hline \end{array}
  &CF &\rightarrow \begin{array}{|c|c|}\hline BD & 0 \\ \hline BDG & CF \\ \hline \end{array}\\
  BDG &\rightarrow \begin{array}{|c|c|}\hline G & CF \\ \hline B & C \\ \hline \end{array}
  &G &\rightarrow \begin{array}{|c|c|}\hline G & CF \\ \hline D & C \\ \hline \end{array}\\
\end{align*}
We deduce that any pattern of the form
\begin{center}\tiny
  \begin{tikzpicture}[node distance=.7cm,>=stealth',bend angle=30,auto,draw]
      \tikzstyle{nice}=[draw,minimum size=.7cm]%
      \node [nice] (a) {CF};%
      \node [nice] (b) [right of=a] {BD};%
      \node [nice] (c) [below of=b] {X};%
      \node [nice] (d) [right of=c] {CF};%
      \node [nice] (e) [right of=d] {BD};%
  \end{tikzpicture}
\end{center}
where X is either BDG or G, is sent to a pattern of the form
\begin{center}\tiny
  \begin{tikzpicture}[node distance=.7cm,>=stealth',bend angle=30,auto,draw]
      \tikzstyle{nice}=[draw,minimum size=.7cm]%
      \node [nice] (euh) {BD};%
      \node [nice] (you) [below of=euh] {BDG};%
      \node [nice] (a) [right of=you] {CF};%
      \node [nice] (b) [right of=a] {BD};%
      \node [nice] (c) [below of=b] {G};%
      \node [nice] (d) [right of=c] {CF};%
      \node [nice] (e) [right of=d] {BD};%
      \node [nice] (cr) [below of=e] {BDG};%
      \node [nice] (dr) [right of=cr] {CF};%
      \node [nice] (er) [right of=dr] {BD};%
  \end{tikzpicture}
\end{center}
Now, observing Figure~(\ref{fig:gamma_steps}), one can see a discrete segment of slope $-\frac{1}{2}$ made of the above pattern starting from the top-left position and reaching the upper-diagonal.  Since we know that all the cells appearing on this discrete segment, namely $G$, $CF$ and $BDG$ (plus an end point that is, depending on the parity of the scale, $B$ or $F$), contain a point of $X_2$, it just remains to show by recurrence that a segment of this form is present at every scale, which is immediate. 
\end{proof}

\begin{ppty}
  No point of $X_2(\Gamma^{-1})$ lies in the interior of the triangle of vertices $(0,\frac{1}{2})$, $(\frac{1}{3},\frac{1}{6})$ and $(\frac{1}{3},\frac{2}{3})$.
\end{ppty}
\begin{proof}
By induction, we can prove that depending on the parity of the step, this triangle takes alternatively the forms presented in Figures~\ref{fig:odd_step} and \ref{fig:even_step}, which represent the corresponding triangle in the substitution system, supposing the pair of initial blocks represents a rectangle of height $1$.  For instance, Figure~\ref{fig:omega_steps},  showing the fifth step, exhibits in this position a triangle of the form presented in Figure~\ref{fig:odd_step}.

\begin{figure}[htb]
\begin{center}\tiny
  \begin{tikzpicture}[node distance=.5cm,>=stealth',bend angle=30,auto,draw]
      \tikzstyle{nice}=[draw,minimum size=.5cm]%
      \node [nice] (euh) {L};%
      \node [nice] (you) [below of=euh] {DH};%
      \node [nice] (a) [right of=you] {EIK};%
      \node [nice] (b) [below of=a] {E};%
      \node [nice] (c) [right of=b] {HL};%
      \node [nice] (d) [below of=c] {DL};%
      \node [nice] (e) [right of=d] {EIK};%
      \node [nice] (cr) [below of=e] {IK};%
      \node [nice] (dr) [right of=cr] {HL};%
      \node [nice] (er) [below of=dr] {DH};%
      \node [nice] (a2) [right of=er] {EIK};%
      \node [nice] (b2) [below of=a2] {E};%
      \node [nice] (c2) [right of=b2] {HL};%
      \node [nice] (d2) [below of=c2] {DL};%
      \node [nice] (e2) [right of=d2] {EIK};%
      \node (S2points1) [below of=e2] {$\ddots$};%
      \node (fake1) [right of=S2points1] {$\ddots$};
      \node (S2points2) [below of=fake1] {$\ddots$};%
      \node (fake2) [right of=S2points2] {$\ddots$};
      \node [nice] (S2end1) [below of=fake2] {DH};%
      \node [nice] (S2end2) [right of=S2end1] {EIK};%
      \node [nice] (S2end3) [below of=S2end2] {E};%
      \node [nice] (S2end4) [right of=S2end3] {D};%
      \node [nice] (S2end5) [below of=S2end4] {H};%
      \node [nice] (S3E1) [above of=S2end4] {DHL};%
      \node [nice] (S3E2) [above of=S3E1] {DHL};%
      \node [nice] (S3E3) [above of=S3E2] {DHL};%
      \node [nice] (S3E4) [above of=S3E3] {DHL};%
      \node [nice] (S3E5) [above of=S3E4] {DHL};%
      \node (S3E6) [above of=S3E5] {\vdots};%
      \node (S3E7) [above of=S3E6] {\vdots};%
      \node (S3E8) [above of=S3E7] {\vdots};%
      \node (S3E9) [above of=S3E8] {\vdots};%
      \node (S3E10) [above of=S3E9] {\vdots};%
      \node (S3E11) [above of=S3E10] {\vdots};%
      \node (S3E12) [above of=S3E11] {\vdots};%
      \node (S3E13) [above of=S3E12] {\vdots};%

      \node [nice] (S1E2) [right of=euh] {K};%
      \node (invisible2) [above of=S1E2] {};
      \node [nice] (S1E3) [right of=invisible2] {H};%
      \node [nice] (S1E4) [right of=S1E3] {K};%
      \node (invisible3) [above of=S1E4] {};
      \node [nice] (S1E5) [right of=invisible3] {H};%
      \node [nice] (S1E6) [right of=S1E5] {K};%
      \node (invisible4) [above of=S1E6] {};
      \node [nice] (S1E7) [right of=invisible4] {H};%
      \node [nice] (S1E8) [right of=S1E7] {K};%
      \node (invisible5) [above of=S1E8] {};%
      \node (S1point1) [right of =invisible5] {\begin{rotate}{27}$\dots$\end{rotate}};
      \node (S1point2) [right of =S1point1] {\begin{rotate}{27}$\cdots$\end{rotate}};
      \node (invisible6) [above of=S1point2] {};%
      \node [nice] (S1E9) [right of=invisible6] {K};%
      \node [nice] (S1E10) [right of=S1E9] {DHL};%
      \node [nice] (S1E11) [above of=S1E10] {DL};%
      \node [nice] (S1E12) [below of=S1E10] {DHL};%
      \node [nice] (S1E13) [below of=S1E12] {DHL};%
  \end{tikzpicture}
\end{center}
\caption{Odd steps (time goes from bottom to top).\label{fig:odd_step}}
\end{figure}

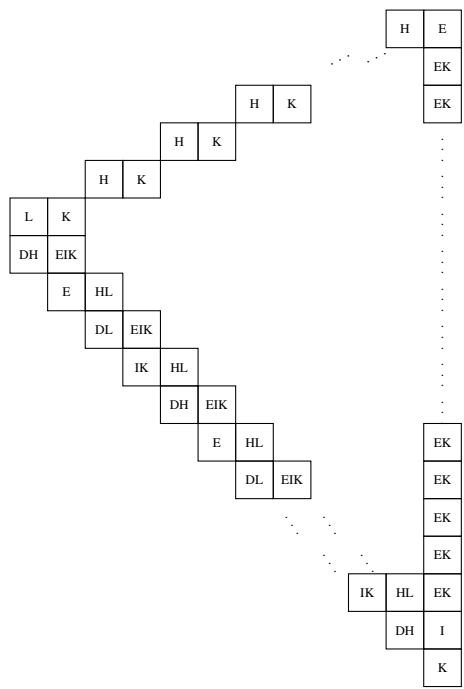
\begin{figure}[htb]
\begin{center}\tiny
  \begin{tikzpicture}[node distance=.5cm,>=stealth',bend angle=30,auto,draw]
      \tikzstyle{nice}=[draw,minimum size=.5cm]%
      \node [nice] (euh) {L};%
      \node [nice] (you) [below of=euh] {DH};%
      \node [nice] (a) [right of=you] {EIK};%
      \node [nice] (b) [below of=a] {E};%
      \node [nice] (c) [right of=b] {HL};%
      \node [nice] (d) [below of=c] {DL};%
      \node [nice] (e) [right of=d] {EIK};%
      \node [nice] (cr) [below of=e] {IK};%
      \node [nice] (dr) [right of=cr] {HL};%
      \node [nice] (er) [below of=dr] {DH};%
      \node [nice] (a2) [right of=er] {EIK};%
      \node [nice] (b2) [below of=a2] {E};%
      \node [nice] (c2) [right of=b2] {HL};%
      \node [nice] (d2) [below of=c2] {DL};%
      \node [nice] (e2) [right of=d2] {EIK};%
      \node (S2points1) [below of=e2] {$\ddots$};%
      \node (fake1) [right of=S2points1] {$\ddots$};
      \node (S2points2) [below of=fake1] {$\ddots$};%
      \node (fake2) [right of=S2points2] {$\ddots$};
      \node [nice] (S2end1) [below of=fake2] {IK};%
      \node [nice] (S2end2) [right of=S2end1] {HL};%
      \node [nice] (S2end3) [below of=S2end2] {DH};%
      \node [nice] (S2end4) [right of=S2end3] {I};%
      \node [nice] (S2end5) [below of=S2end4] {K};%
      \node [nice] (S3E1) [above of=S2end4] {EK};%
      \node [nice] (S3E2) [above of=S3E1] {EK};%
      \node [nice] (S3E3) [above of=S3E2] {EK};%
      \node [nice] (S3E4) [above of=S3E3] {EK};%
      \node [nice] (S3E5) [above of=S3E4] {EK};%
      \node (S3E6) [above of=S3E5] {\vdots};%
      \node (S3E7) [above of=S3E6] {\vdots};%
      \node (S3E8) [above of=S3E7] {\vdots};%
      \node (S3E9) [above of=S3E8] {\vdots};%
      \node (S3E10) [above of=S3E9] {\vdots};%
      \node (S3E11) [above of=S3E10] {\vdots};%
      \node (S3E12) [above of=S3E11] {\vdots};%
      \node (S3E13) [above of=S3E12] {\vdots};%

      \node [nice] (S1E2) [right of=euh] {K};%
      \node (invisible2) [above of=S1E2] {};
      \node [nice] (S1E3) [right of=invisible2] {H};%
      \node [nice] (S1E4) [right of=S1E3] {K};%
      \node (invisible3) [above of=S1E4] {};
      \node [nice] (S1E5) [right of=invisible3] {H};%
      \node [nice] (S1E6) [right of=S1E5] {K};%
      \node (invisible4) [above of=S1E6] {};
      \node [nice] (S1E7) [right of=invisible4] {H};%
      \node [nice] (S1E8) [right of=S1E7] {K};%
      \node (invisible5) [above of=S1E8] {};%
      \node (S1point1) [right of =invisible5] {\begin{rotate}{27}$\dots$\end{rotate}};
      \node (S1point2) [right of =S1point1] {\begin{rotate}{27}$\cdots$\end{rotate}};
      \node (invisible6) [above of=S1point2] {};%
      \node [nice] (S1E9) [right of=invisible6] {H};%
      \node [nice] (S1E10) [right of=S1E9] {E};%
      \node [nice] (S1E12) [below of=S1E10] {EK};%
      \node [nice] (S1E13) [below of=S1E12] {EK};%
  \end{tikzpicture}
\end{center}
\caption{Even steps (time goes from bottom to top).\label{fig:even_step}}
\end{figure}

The proof that each of the figures substitutes into the other one is purely mechanical, and essentially done by the very existence of Figure~\ref{fig:omega_steps}, where the first five steps of substitution are readable.
\end{proof}

\section{Discussion}
We gave a new necessary condition for the simulation of CA and applied it to solve a few open questions of the form `Does there exist a reversible CA that simulates such and such but not such and such?'.  Noticeably, each time we were able to answer this question, it was in the positive, which is one general reason why we would expect the same answer for other closely related questions of the same sort that remain open.

Our method is tailored to be applied to linear CA.  Their practical advantage is that much of the information is present in their spacetime diagram, and therefore easy to access and comprehend.  For instance, with our theorem in mind, a blink at Figure~\ref{fig:theta} is enough to suspect that $\Theta$ cannot simulate the identity.  It then remains to check rigorously that the pattern does represent $X_2$ accurately, but that part is purely mechanical, if a bit tedious.  Let us now finish with two questions.

Why did the authors resort to a $3\times 3$ matrix?  Couldn't they find anything simpler?  No, they could not.  Actually they conjecture that every $2\times 2$ matrix simulates its inverse, which interestingly enough reduces to deciding whether every matrix simulates its transpose.

Does there exist a CA that can simulate the identity, but not its inverse/dual?  The correct answer is `probably, and $\Gamma\times\id$ is a good candidate'.  However, our theorem is not really helpful in this case, since the $X_p$-s of this CA are trivial.  Hopefully some hybrid can be created by merging it with \cite[theorem 3.4]{bulking2} and made available to the masses in the future.

\bibliographystyle{alpha}
\bibliography{biblio}

\end{document}